\documentclass[preprint,12pt]{elsarticle}

\usepackage{amssymb}
\usepackage{amsmath}
\usepackage{amsthm}
\usepackage{subfig}
\usepackage{graphicx}
\usepackage{multirow} 
\usepackage{booktabs}
\usepackage{array}
\newtheorem{thm}{Theorem}
\newtheorem{definition}{Definition}

\newtheorem{assumption}{Assumption}
\newtheorem{lemma}{Lemma}
\newtheorem{problem}{Problem}
\newtheorem{remark}{Remark}

\journal{Journal of the Franklin Institute}

\begin{document}

\begin{frontmatter}



\title{HJB-based online safety-embedded critic learning for uncertain systems with self-triggered mechanism}

\author[label1]{Zhanglin Shangguan, Bo Yang, Qi Li}
\author[label2]{Wei Xiao}
\author[label1]{Xingping Guan}
\affiliation[label1]{organization={School of Automation and Intelligent Sensing, Shanghai Jiao Tong University},
            city={Shanghai},
            postcode={200240},
            country={China}}

\affiliation[label2]{organization={Computer Science and Artificial Intelligence Laboratory, Massachusetts Institute of Technology},
            city={Cambridge},
            postcode={02139},
            state={Massachusetts},
            country={USA}}



\begin{abstract}
This paper presents a learning-based optimal control framework for safety-critical systems with parametric uncertainties, addressing both time-triggered and self-triggered controller implementations. First, we develop a robust control barrier function (RCBF) incorporating Lyapunov-based compensation terms to rigorously guarantee safety despite parametric uncertainties. Building on this safety guarantee, we formulate the constrained optimal control problem as the minimization of a novel safety-embedded value function, where the RCBF is involved via a Lagrange multiplier that adaptively balances safety constraints against optimal stabilization objectives. To enhance computational efficiency, we propose a self-triggered implementation mechanism that reduces control updates while maintaining dual stability-safety guarantees. The resulting self-triggered constrained Hamilton-Jacobi-Bellman (HJB) equation is solved through an online safety-embedded critic learning framework, with the Lagrange multiplier computed in real time to ensure safety. Numerical simulations demonstrate the effectiveness of the proposed approach in achieving both safety and control performance.
\end{abstract}


\end{frontmatter}



\section{Introduction}
\label{sec1}
Modern control systems face significant challenges in safety-critical applications like robotics, aerospace, and power systems, where nonlinear dynamics, state constraints, and model uncertainties demand advanced control strategies \cite{hsu2023safety}. Traditional methods often struggle to ensure both performance and stability in such complex systems, while practical implementations must additionally contend with stringent computational and communication constraints. These challenges underscore the need for novel control frameworks that simultaneously address performance optimization, safety guarantee, and resource-aware operation in uncertain environments.

In recent years, learning-based optimal control methods have emerged as promising solutions for control performance optimization, particularly within the framework of adaptive dynamic programming (ADP) \cite{lewis2013reinforcement}. ADP has demonstrated significant potential in stabilizing nonlinear systems by addressing the Hamilton-Jacobi-Bellman (HJB) equation. Central to ADP is the actor-critic architecture \cite{vamvoudakis2010online}, where the actor implements a control policy and the critic evaluates the associated costs, providing feedback through reward or punishment signals. However, this dual architecture appears redundant, as the control policy can be derived directly from the value function. As a result, a streamlined approach using an only-critic framework has been proposed, improving the efficiency of the learning process \cite{wang2017adaptive,yang2020event}. This only-critic framework has been further extended to the identifier-critic (IC) framework, enabling online learning for systems with parametric uncertainties \cite{na2020adaptive}. However, the aforementioned IC learning-based optimal controllers are typically designed without accounting for state constraints, limiting their applicability in safety-critical systems. To address the challenges posed by safety constraints, recent advances in ADP-based optimal control have prioritized the strict enforcement of hard state constraints—both during learning and deployment.

To ensure that system states are constrained in a user-defined safe set, state constraints are considered in ADP techniques by incorporating reciprocal barrier functions in value functions \cite{cohen2020approximate,greene2020sparse}. However, this method requires the assumption that the augmented value function maintains Lipschitz continuity, which may fail to hold for certain systems and safety constraints. Moreover, control performance learning and safety guarantees are coupled without rigorous safety guarantees. To strictly enforce safety and decouple safety from learning, control barrier functions (CBFs) \cite{ames2016control} are utilized to modify the nominal optimal controller \cite{cheng2019end,peng2023design}. While safety can be strictly guaranteed through CBF-based safety filters, these safe controllers only achieve point-wise optimality, and the inverse optimal properties of CBF-based safety filters are detailed in the work \cite{krstic2023inverse}. To overcome the myopia of CBF-based safety filters, \cite{cohen2023safe} proposes a safe exploration technique that hybridizes the CBF-filtered safe control policy and the optimal stabilization-oriented control policy to learn the optimal value function. This approach embeds safety constraint information into the critic learning process of the optimal value function. CBF-based safety constraints are synthesized in HJB equations, which are solved offline to achieve optimal stabilization and safety objectives simultaneously \cite{almubarak2021hjb}. We observe that the works \cite{cohen2023safe,wang2025learning} employ constant weighting factors to balance controller safety and optimal stabilization. In contrast, our approach, similar to \cite{bandyopadhyay2023hjb}, embeds CBF-based safety constraints into the value function through adaptive Lagrange multipliers. Our work further advances this framework by: incorporating robust safety constraint formulation under parametric uncertainties and decoupling safety guarantees from control performance learning.

In systems with limited computational and communication resources, ADP-based optimal controllers are often implemented with event-triggered mechanisms \cite{vamvoudakis2014event,vamvoudakis2018model}. However, the event-triggered mechanism designed in \cite{vamvoudakis2014event,vamvoudakis2018model} only considers the stability of the system, and the CBF-based safeguard controller may lose safety guarantees during the triggering intervals. With a periodic sampling mechanism, baseline CBFs are enhanced with compensation terms related to the sampling period to guarantee the safety properties of systems using periodic sampling controllers \cite{breeden2021control,sun2024safety}. The introduction of the compensation term makes the CBF more conservative. In fact, when safety constraints are inactive, the system inherently possesses a certain ``safety margin", allowing an event-triggered CBF calculation mechanism without requiring the compensation term by leveraging this ``safety margin" \cite{yang2019self}. The work \cite{xiao2022event} further extends this framework to systems with unknown dynamics, where event-triggered CBFs are applied to auxiliary adaptive dynamics. The work \cite{sabouni2024optimal} investigates the application of event-triggered CBFs in the safety control of connected automated vehicles, effectively alleviating computational and communication burdens. However, event-triggered control performance and safety guarantee in \cite{yang2019self,xiao2022event,sabouni2024optimal} are synthesized in a point-wise optimization framework. Aligned with \cite{xiao2022event}, our research interest lies in further extending this methodology to an infinite-horizon optimal control framework.

Compared with the previous literature, the contributions of this paper are threefold. 
\begin{enumerate}
\item We establish a theoretical framework for safety-critical optimal control by deriving a HJB equation that explicitly incorporates safety constraints through robust CBFs (RCBFs). Building on input-to-state safe (ISSf)-CBF \cite{kolathaya2018input}, our approach uniquely integrates parameter estimation dynamics: if an identifier guarantees convergence, we design precise compensation terms to maintain safety under model uncertainties. Furthermore, by leveraging Karush-Kuhn-Tucker (KKT) conditions, we embed RCBF constraints into HJB equations of infinite-horizon optimal control problems, formally unifying control performance and safety guarantees.
\item We develop a novel self-triggered mechanism to efficiently implement the RCBF-constrained optimal controller with dual stability-safety guarantees. Our approach shares a similar design philosophy with \cite{yang2019self,xiao2022event} while extending the framework to infinite-horizon optimal control. By introducing the concept of ``safety margin", we establish a dual-threshold triggering mechanism that automatically shifts between stability-dominant triggering (during inactive constraint periods) and safety-critical triggering (near constraint boundaries). This formulation effectively prevents ``Zeno behavior" and reduces the computation and communication overhead. 
\item We propose an online safety-embedded critic learning framework to approximate solutions to constrained HJB equation. Key innovations include: (i) a parameter identifier based on filtered auxiliary integral variables is utilized, where a ``refresh" mechanism is introduced to enforce convergence of the estimated parameters to their true values, enabling accurate RCBF compensation; and (ii) unlike work \cite{cohen2023safe,wang2025learning} that uses fixed safety-stability trade-off gains, we introduce a real-time calculated Lagrange multiplier to adaptively balance optimal stabilization and safety during critic weight learning.
\end{enumerate}

\textbf{Notations.} A continuous function $\alpha:\mathbb{R}_{\ge 0}\mapsto\mathbb{R}_{\ge 0}$ is classified as a class $\mathcal{K}_{\infty}$ function provided that it is strictly increasing and $\alpha(0)=0,\ \lim_{r\mapsto\infty}\alpha(r)=\infty$. The symbol $\Vert \cdot \Vert$ signifies the standard Euclidean norm. For a compact set $\mathcal{X}\subset\mathbb{R}^n$ and a continuous function $(\cdot):\mathcal{X}\mapsto \mathbb{R}^N$, let $\overline{\Vert (\cdot) \Vert}$ be defined as $\sup_{x\in\mathcal{X}} \Vert(\cdot)\Vert$. The notation $\lambda_{\max}(M)$, $\lambda_{\min}(M)$ represents the maximum and minimum eigenvalues of matrix $M$, respectively. $\mathbb{I}_n \in \mathbb{R}^{n \times n}$ stands for an $n \times n$ identity matrix. A closed ball centered at $x_0\in\mathbb{R}^n$ with radius $r>0$ is denoted by $\mathcal{B}_r(x_0)$, which is defined as $\{x\in\mathbb{R}^n | \Vert x-x_0 \Vert\leq r \}$. 

\section{Preliminaries and problem formulations}
We investigate a class of uncertain nonlinear systems governed by:
\begin{flalign}\label{system}
& \dot{x}(t) = \zeta(x,\theta) + \rho(x)u, 
\end{flalign} 
where $x\in\mathbb{R}^n$ denotes the state vector, $\theta\in\Theta\subset\mathbb{R}^p$ represents parametric uncertainties, and $u\in\mathbb{R}^m$ is the control input vector. The functions $\zeta:\mathbb{R}^n \times \mathbb{R}^p \mapsto \mathbb{R}^n$, and $\rho:\mathbb{R}^n \to \mathbb{R}^{n\times m}$, characterize the nonlinear drift dynamics and the input influence matrix, respectively. We consider the case in which the unknown drift dynamic $\zeta(x,\theta)$ can be expressed as a linear combination of a known regressor matrix $\omega:\mathbb{R}^n\to\mathbb{R}^{n\times p}$ and an unknown parameter vector $\theta$, which is $\zeta(x,\theta)=\omega(x)\theta$. To ensure well-posedness of solutions, we require the designed control policy $u(t)$ to render the system dynamics, $(\zeta,\varrho)$, locally Lipschitz continuous in $x$ and piecewise continuous in $t$. Under these conditions, the trajectory of the system \eqref{system} is uniquely determined for any initial state $x(t_0)=x_0$. 
 
\subsection{Safety guarantees via control barrier functions}
System safety is quantified by the state trajectory remaining within a prescribed safe operating region defined by: $\mathcal{D} = \{ x\in\mathcal{X} : s(x)\ge 0\}$, where $s:\mathcal{X}\mapsto\mathbb{R}$ is a continuously differentiable safety certification function. A valid safe controller must ensure the forward invariance property: $x(0)\in \mathcal{D}\Longrightarrow x(t)\in\mathcal{D},\ \forall t\ge t_0$.

A highly effective method for designing controllers that make the set $\mathcal{D}$ forward invariant is the control barrier functions (CBFs):
\begin{definition}
\cite{ames2016control} For system \eqref{system} and user-defined safe set $\mathcal{D}$, a continuously differentiable function $s:\mathcal{X}\to\mathbb{R}$ qualifies as a CBF if it satisfies: 
\begin{flalign}\label{3}
& \sup_{u\in\mathcal{U}}\{\mathcal{L}_{\zeta} s(x) + \mathcal{L}_{\rho} s(x)u \}\ge -\alpha(s(x)) 
\end{flalign}
for all $x\in\mathcal{X}$, where $\alpha$ belongs to the class of extended $\mathcal{K}_{\infty}$ functions, $\mathcal{L}_{\zeta}$ and $\mathcal{L}_{\rho}$ represent Lie derivatives along the system dynamics.
\end{definition}

When complete system knowledge is available, the CBF condition generates a family of safe control laws by imposing instantaneous constraints on admissible control inputs. This methodology provides a systematic approach to safety-critical controller synthesis for general nonlinear systems.

\subsection{Constrained optimal control problem}
With the user-defined safe set $\mathcal{D}$, we focus on the following infinite-horizon constrained optimal control problem for the uncertain system \eqref{system}:
\begin{problem} Constrained optimal control problem (COCP)
\begin{subequations}
\begin{flalign}
\min_{u\in\mathcal{U}} \mathcal{J}(x,u)&=\int_{t_0}^{\infty} \underbrace{\left(x(\tau)^{\top}Qx(\tau) + \frac{1}{2}u(\tau)^{\top}Ru(\tau) \right)}_{l(x(\tau),u(\tau))} d\tau \label{value function}\\
 {\rm s.t.}\quad \dot{x} &= \zeta(x,\theta) + \rho(x)u, \\
x&(t)  \in \mathcal{D}\quad \forall t\ge t_0,
\end{flalign}
\end{subequations}
where $Q\in\mathbb{R}^{n\times n}$ is a state weighting matrix $(Q=Q^\top \succ 0)$, $R\in\mathbb{R}^{m\times m}$ is an input weighting matrix $(R=R^\top\succ 0)$, $\mathcal{J}(x,u)$ is the infinite-horizon performance metric, and $l(x,u)$ is the instantaneous cost function.
\end{problem}\label{cocp}

The COCP above presents four key challenges for controller design: (a) Ensuring robust safety under system dynamics with model uncertainties, where baseline CBF methods \eqref{3} may lose safety guarantee; (b) Achieving optimal stabilization while maintaining safety for nonlinear systems, requiring careful trade-off analysis; (c) Developing event-triggered execution mechanisms to alleviate computational and communication burdens in resource-constrained environments; (d) Designing online learning frameworks capable of simultaneously estimating model uncertainties and approximating the solution controller for the COCP.

\section{RCBF-constrained HJB equations}\label{sec3}
In this section, we address the first two challenges posed by the aforementioned control objectives: By introducing a compensation term into the baseline CBF, we design an RCBF to avoid potential safety violations caused by model uncertainties. By incorporating Lagrange multipliers for the RCBF constraints, we bridge safety and optimal stabilization, deriving the RCBF-constrained HJB equation, that is, the condition that the RCBF-induced optimal safe controller must satisfy.

\subsection{Encoding safety via RCBFs}
Model uncertainties can lead to failure of the baseline CBF-based safety constraint \eqref{3}. To ensure the safety of systems under model uncertainties, we concentrate on the design of RCBF-based safety constraints. We denote $\theta$ = $\hat{\theta}$ + $\tilde{\theta}$, where $\hat{\theta}$ is an estimate of uncertain parameters and $\tilde{\theta}$ is the error between the estimate value and the true value. We assume that $\Vert \tilde{\theta} \Vert_{\infty} \leq \overline{\theta}$. Similarly to \cite{kolathaya2018input}, treating $\tilde{\theta}$ as a bounded external disturbance signal, we can define that: for the controller $u$ and $\alpha_1,\alpha_2\in\mathcal{K}_{\infty}$, if $s$ satisfies $\mathcal{L}_\omega s(x)\hat{\theta} + \mathcal{L}_{\rho} s(x) u(x) \ge -\alpha_1 \big(s(x)\big) - \alpha_2 \big( \Vert \tilde{\theta} \Vert_{\infty} \big)$, then $s(x)$ is an input-to-state safe barrier function (ISSf-BF). The existence of an ISSf-BF implies that system \eqref{system} is ISSf with respect to the safe set $\mathcal{D}$, meaning that there exists some $\alpha_3\in\mathcal{K}_{\infty}$ such that the set $\overline{\mathcal{D}} = \left\{x: s(x) + \alpha_3\big( \Vert \tilde{\theta} \Vert_{\infty} \big)\ge 0 \right\}$ is forward invariant. The ISSf property quantifies the impact of external disturbances on system safety. From a controller design perspective, a safety compensation term related to the ISSf gains $\alpha_2,\alpha_3$ should be incorporated into the controller to avoid potential safety violations caused by $\tilde{\theta}$.
\begin{lemma}
\cite{kolathaya2018input} Augment the baseline CBF in \eqref{3} with a compensation term $\Xi(x)$, yielding the following RCBF-based safety constraint:
\begin{flalign}\label{rcbf}
\mathcal{L}_\omega s(x)\hat{\theta} + \mathcal{L}_{\rho} s(x) u(x) \ge -\alpha \big(s(x)\big) + \Xi(x),
\end{flalign}
where $\Xi(x) = \varpi\Vert \mathcal{L}_\omega s(x) \Vert^2$ with $\varpi > 0$ being a user-defined constant. Any controller $u$ satisfying \eqref{rcbf} guarantees the ISSf of the safe set $\mathcal{D}$ w.r.t. $\tilde{\theta}$.
\end{lemma}
The above analysis treats $\tilde{\theta}$ merely as a generic external disturbance. In practice, if an identifier is available to adaptively estimate the uncertain parameter $\theta$, we can design a more precise compensation term $\Xi(x)$ - that is, a more accurate $\varpi$. Assume that there exists an exponentially convergent identifier with a Lyapunov function $V_{\theta}(\tilde{\theta})$ satisfying:
\begin{flalign}\label{theta lyapunov}
k_1 \Vert \tilde{\theta} \Vert^2 \leq V_{\theta}(\tilde{\theta},t) \leq k_2 \Vert \tilde{\theta} \Vert^2, \quad
\dot{V}_{\theta}(\tilde{\theta},t) \leq -k_3 V_{\theta}(\tilde{\theta},t),
\end{flalign}
where $k_1, k_2, k_3 > 0$. 

Motivated by \cite{taylor2020adaptive}, we consider a robust safe set augmented with this Lyapunov function $V_{\theta}$: $\mathcal{D}_{\theta}\triangleq\{x\in\mathcal{X}:s_{\theta}(x,t)\ge0\}$, where $s_{\theta}(x,t)\triangleq s(x) - \eta V_{\theta}$ with $\eta>0$ being a user-defined constant. It is evident that the safe set $\mathcal{D}_{\theta}$ is more conservative than the original safe set $\mathcal{S}$. Furthermore, as time approaches infinity, $\mathcal{S}_{\theta}$ asymptotically converges to $\mathcal{S}$.
\begin{thm}\label{th-cbf1}
Suppose that $s_{\theta}(x(0),0)>0$ and $\frac{\partial s}{\partial x}\neq 0$ for all $x\in\mathcal{D}$. Given an identifier with a Lyapunov function $V_{\theta}$ characterizing the convergence of the estimation error $\tilde{\theta}$, we can derive a precisely compensated RCBF \eqref{rcbf} by selecting: (i) the scaling factor $\varpi = \frac{1}{4(\eta-\eta_c)k_1k_3}$, and (ii) the class $\mathcal{K}_{\infty}$ function $\alpha(s) = \frac{\eta_c k_3}{\eta}s$, where the design parameters satisfy $\eta > \eta_c > 0$.
\end{thm}
\begin{proof}
Taking the time derivative of $s_{\theta}(x,t)$ yields
\begin{flalign}\label{rcbf-proof}
\dot{s}_{\theta}&(x,t) = \dot{s}(x) - \eta \dot{V}_{\theta} \notag \\
&\ge \mathcal{L}_\omega s(x)(\hat{\theta}+\tilde{\theta}) + \mathcal{L}_{\rho} s(x) u(x) + (\eta - \eta_c) k_3 V_{\theta} + \eta_c k_3 V_{\theta} \notag \\
&\ge \mathcal{L}_\omega s(x)\hat{\theta} + \mathcal{L}_{\rho}s(x)u - \Vert \mathcal{L}_\omega s(x) \Vert \Vert \tilde{\theta} \Vert  +(\eta - \eta_c)k_1k_3\Vert \tilde{\theta} \Vert^2 + \eta_c k_3V_{\theta} 
\end{flalign}
Applying the RCBF condition \eqref{rcbf} to the above inequality leads to:
\begin{flalign}\label{rcbf proof2}
\dot{s}_{\theta}(x,t)\ge &-\alpha(s)  + \eta_c k_3V_{\theta} + \left( \sqrt{\varpi}\Vert \mathcal{L}_\omega s(x) - \frac{\Vert \tilde{\theta} \Vert}{2\sqrt{\varpi}} \Vert   \right)^2\notag\\
& - \frac{\Vert \tilde{\theta} \Vert^2}{4\varpi} + (\eta - \eta_c)k_1k_3\Vert \tilde{\theta} \Vert^2,    
\end{flalign}
Substituting $\varpi = \frac{1}{4(\eta-\eta_c)k_1k_3}$ and $\alpha(s) = \frac{\eta_c k_3}{\eta}s(x)$ into \eqref{rcbf proof2} yields $\dot{s}_{\theta}(x,t) \ge -\frac{\eta_c k_3}{\eta} s_{\theta}(x,t)$. Therefore, the controller $u(x)$ is verified to render the set $\mathcal{S}_{\theta}$ forward invariant when RCBF \eqref{rcbf} holds. Note that $\mathcal{S}_{\theta}\subset\mathcal{S}$, which means that the set $\mathcal{S}$ is forward invariant.
\end{proof}

\subsection{Optimal controller in RCBF-constrained policy space}
Recall the COCP in \ref{cocp}, where the RCBF \eqref{rcbf} enforces the forward invariance of the safe set $\mathcal{D}$. We now explore how to derive a safe control policy that minimizes the performance metric $J(x,u)$ within the RCBF-constrained policy space.

Let $V_{\lambda}^*(x)$ denote the ``safety-embedded" value function (SVF), representing the optimal cost-to-go under the RCBF-constrained policy space: $V_{\lambda}^*(x) = \inf_{u\in\mathcal{U}_{\rm rcbf}}\mathcal{J}(x,u)$. The RCBF-based safety constraint is given by: $\nu(x,\theta,u) = \mathcal{L}_\omega s(x)\theta + \mathcal{L}_{\rho} s(x)u+ \alpha(s) - \Xi(x)\ge 0$. By introducing a Lagrange multiplier $\lambda\left(x(t)\right) \ge 0\ \ \forall t\in[t_0,\infty)$ for this constraint, we can formulate the SVF $V_{\lambda}^*(x)$ as 
\begin{flalign}\label{lagrange value function}
    V_{\lambda}^*(x) = \min_{u}\max_{\lambda}\int_{t}^{\infty}\Big( l(x,u)-\lambda\cdot\nu(x,\theta,u) \Big){\rm d}\tau.
\end{flalign}
Based on the fundamental principle of dynamic programming and the KKT conditions \cite{almubarak2021hjb}, we can obtain the stationary RCBF-constrained Hamilton-Jacobi-Bellman (HJB) partial differential equation \cite{kirk2004optimal}: for any $x\in\mathcal{X}$,
\begin{flalign}\label{constrained HJB}
\min_{u}\max_{\lambda\ge 0} \Big\{ & l(x,u) -\lambda\cdot\nu(x,\theta,u) + \nabla V_{\lambda}^{*}\big[ \zeta(x,\theta) + \rho(x)u \big] \Big\}= 0, 
\end{flalign}
where the solution $\left(\lambda^*,u_{\lambda}^*\right)$ satisfies the complementary slackness condition $\lambda^*\left(x\right)\cdot \nu(x,\theta,u_{\lambda}^*)=0$. Since the system is autonomous and the RCBF-based safety constraint is stationary, the Lagrange multiplier is a function of the state $x$. For any fixed state $x$, the constrained HJB \eqref{constrained HJB} is convex with respect to the control policy $u$, allowing us to derive the optimal control policy $u_{\lambda}^*$ as
\begin{flalign}\label{optimal safe controller}
& u_{\lambda}^*\big(x,\lambda^*(x)\big) = R^{-1}\big[ \lambda^*(x)\mathcal{L}_{\rho} s(x)^\top - \mathcal{L}_{\rho}V_{\lambda}^{*}(x)^\top\big]. 
\end{flalign}
Substituting \eqref{optimal safe controller} into $\lambda^*\left(x\right)\cdot \nu(x,\theta,u_s^*)=0$, we can obtain the optimal Lagrange multiplier as
\begin{flalign}\label{optimal lamda}
& \lambda^*(x) = \max\left( -\frac{\nu(x,\theta,u_{\rm no}(x))} {R_{s\rho}(x)}, 0\right), 
\end{flalign}
where $u_{\rm no}(x) = -R^{-1}\mathcal{L}_{\rho}V_{\lambda}^*(x)^{\top}$, $R_{\rho}(x) = \rho(x)R^{-1}\rho(x)^\top$, and $R_{s\rho}(x) = \nabla s(x) R_{\rho}(x) \nabla s(x)^\top$.

\begin{remark}
The term $\nu(x,\theta,u_{\rm no}(x))$ is defined with the nominal control law $u_{\rm no}(x) = -R^{-1}L_{\rho}V_{\lambda}^*(x)^{\top}$, independent of the safety correction term $\lambda^* L_{\rho} s$. The term $\nu(x,\theta,u_{\rm no}(x))$ can be seen as a ``safety margin" at the current state $x$. If $\nu(x,\theta,u_{\rm no}(x))$ is positive for a specific state $x$, then the nominal control law $u_{\rm no}$ is adequate to guarantee system safety at $x$, rendering the safety correction term zero. Conversely, when $\nu(x,\theta,u_{\rm no}(x))$ is negative, the safety correction term exerts a non-zero effectiveness to ensure that the RCBF-based safety constraint is satisfied at state $x$. Robustness to parametric uncertainties is reflected in the larger Lagrange multipliers $\lambda^*$ and the value of the Lagrange multipliers balances the requirements of safety and optimal stabilization for the controller.
\end{remark} 

\section{Self-triggered mechanism}
In this section, we address the third challenge posed by the aforementioned control objectives: we design a self-triggered mechanism to implement the RCBF-constrained optimal controller $u_{\lambda}^*(x)$ \eqref{optimal safe controller}. This self-triggering mechanism is designed to simultaneously ensure stability and safety during triggering intervals while reducing computational overhead and communication frequency. Moreover, the designed self-triggered mechanism determines the next triggering instant based on the current sampled state, eliminating the need for continuous monitoring of the system. Here, we begin with the fundamental framework of event-triggered implementation mechanism.

The event-triggered control scheme is implemented through a discrete-time execution framework. This approach establishes an ordered sequence of activation time instants $\{t_j \}_{j=0}^{\infty}$ with $t_0=0$ and $t_j < t_{j+1}$ for all $j\ge 0$, where each triggering interval $T_j = t_{j+1}-t_j$ represents the duration between consecutive control updates. We denote the sampled state as
\begin{flalign}
&\breve{x}(t) = \left\{ \begin{array}{ll}
	x(t_j),& t\in[t_j,t_{j+1})  \\
	x(t_{j+1}), & t=t_{j+1}. 
\end{array} \right.
\end{flalign}  
As analyzed in the previous section, the time-triggered optimal safe controller with a Lagrange multiplier function $\lambda^*(x)$ takes the form of \eqref{optimal safe controller}. During triggering intervals, the controller maintains constant via a zero-order holder of the most recent sampled state measurement $\breve{x}(t_j)$ throughout each interval $[t_j,t_{j+1})$, only acquiring fresh state information when the triggering condition is violated at $t_{j+1}$. Based on this, we can formulate the event-triggered controller as:
\begin{flalign}\label{et-controller}
& u_d^* \big(\breve{x},\lambda^*(\breve{x})\big) = \left\{ \begin{array}{ll}
	u^*_{\lambda}\Big(x(t_j),\lambda^*(x(t_j))\Big), &t\in[t_j,t_{j+1})  \\
	u^*_{\lambda}\Big(x(t_{j+1}),\lambda^*(x(t_{j+1}))\Big), & t=t_{j+1}. 
\end{array} \right.
\end{flalign}
The event-sampling error $\breve{e}_j(t)=x(t)-\breve{x}_j,t\in[t_j,t_{j+1})$ quantifies the discrepancy between the true state $x$ and the sampled state $\breve{x}$.

\subsection{Triggering threshold for stability guarantee}
Before presenting the triggering threshold for stability guaranteeing, we make the following assumption regarding the ``nominal" controller:
\begin{assumption}\label{lipschitz}
The nominal optimal control policy $u_{\rm no} = R^{-1}\mathcal{L}_{\rho}V^*_{\lambda}(x)^\top$ is $d_v$-Lipschiz:
\begin{flalign}
    \left\Vert \mathcal{L}_{\rho}V^*_{\lambda}(\breve{x})R^{-1} - \mathcal{L}_{\rho}V^*_{\lambda}(x)R^{-1} \right\Vert \leq d_v \Vert \breve{x}-x \Vert = d_v\Vert \breve{e} \Vert, 
\end{flalign}
where $d_v$ is the Lipschitz constant of the function $\mathcal{L}_{\rho}V^*_{\lambda}(\cdot)R^{-1}$.
\end{assumption}
\begin{remark}
The time-triggered controller \eqref{optimal safe controller} consists of two parts: one is the so-called ``nominal" controller \( u_{\rm no} \), and the other is the safety-inducing term \( R^{-1}\lambda^*(x)\mathcal{L}_{\rho} s(x)^\top \). Intuitively, the design of the triggering threshold related to stability is only associated with the ``nominal" controller, since there always exists a neighborhood centered at the equilibrium point $\mathcal{B}_r(0)$ where the constraints are not activated and the safety-inducing term is zero.
\end{remark}
To ensure system stability during triggering intervals, we propose the following triggering threshold condition concerning event-sampling error:
\begin{flalign}\label{etm}
& \Vert \breve{e}_j(t) \Vert \leq \sqrt{\frac{(1-\mu)\lambda_{\min}(Q)\Vert x \Vert^2}{2d_v^2\Vert R \Vert} } \triangleq f_{\rm v,event}(\Vert x\Vert),\quad t\in[t_{j},t_{j+1}) 
\end{flalign}
where $f_{\rm v,event}(\Vert x\Vert)$ is the event-triggering threshold function.
\begin{thm}\label{etm thm}
Suppose that Assumption \ref{lipschitz} holds. We can conclude that the true state of the system under the event-triggered controller is UUB, provided that the event-triggering condition \eqref{etm} holds.
\end{thm}
\begin{proof}\label{proof for st threshold}
Let the optimal SVF $V_{\lambda}^*(x)$ be the Lyapunov candidate. Differentiating $V_{\lambda}^*(x)$ along the solution of $\dot{x}=\zeta(x,\theta)+\rho(x)u_{d}(\breve{x})$, we have
\begin{flalign}
\dot{V}_{\lambda}^* &= \nabla V_{\lambda}^*\Big(\zeta(x,\theta)+\rho(x)u_{d}^*\big(\breve{x},\lambda^*(\breve{x})\big) \Big) \notag\\[3pt]
&= \nabla V_{\lambda}^*\Big(\zeta(x,\theta)+\rho(x) u_{\lambda}^*\big(x,\lambda^*(x)\big) \Big) + \nabla V_{\lambda}^* \rho \Big( u_{d}^*\big(\breve{x},\lambda^*(\breve{x})\big)-u_{\lambda}^*\big(x,\lambda^*(x)\big) \Big) \notag \\[3pt]
&= -x^TQx - \frac{1}{2}u_{\lambda}^{*\top}Ru_{\lambda}^* + \nabla V_{\lambda}^* \rho \Big( u_{d}^*\big(\breve{x},\lambda^*(\breve{x})\big)-u_{\lambda}^*\big(x,\lambda^*(x)\big) \Big). \label{556}  
\end{flalign}
Note that $\nabla V_{\lambda}^*\rho = \lambda^*\mathcal{L}_{\rho} s - u_{\lambda}^*\big(x,\lambda^*(x)\big)^{\top} R$. We can obtain that
\begin{flalign}\label{333}
\nabla V_{\lambda}^*\rho\Big( u_{d}^*\big(\breve{x},\lambda^*(\breve{x})\big) &- u_{\lambda}^*\big(x,\lambda^*(x)\big) \Big) =\underbrace{\lambda^* \mathcal{L}_{\rho} s\Big( u_{d}^*\big(\breve{x},\lambda^*(\breve{x})\big)-u_{\lambda}^*\big(x,\lambda^*(x)\big) \Big)}_{\Upsilon_1} \notag\\
&\underbrace{ -u_{\lambda}^*\big(x,\lambda^*(x)\big)^{\top}R\big( u_{d}^*\big(\breve{x},\lambda^*(\breve{x})\big)-u_{\lambda}^*\big(x,\lambda^*(x)\big) \big) }_{\Upsilon_2}. 
\end{flalign}
Using Young's inequality, we have that
\begin{flalign}
& \Upsilon_1 \leq \frac{1}{2}\left\Vert \lambda^* \mathcal{L}_{\rho} s R^{-\frac{1}{2}}\right\Vert^2 + \frac{1}{2}\left\Vert R \right\Vert \left\Vert u_{d}^*\big(\breve{x},\lambda^*(\breve{x})\big) - u_{\lambda}^*\big(x,\lambda^*(x)\big) \right\Vert^2, \label{554}\\
& \Upsilon_2 \leq \frac{1}{2}u_{\lambda}^{*\top}Ru_{\lambda}^* + \frac{1}{2}\Vert R \Vert \left\Vert u_{d}^*\big(\breve{x},\lambda^*(\breve{x})\big) - u_{\lambda}^*\big(x,\lambda^*(x)\big) \right\Vert^2. \label{555}
\end{flalign}
Using the fact \eqref{optimal safe controller} and the conclusion of Assumption \ref{lipschitz}, we can obtain that
\begin{flalign}\label{mkl}
&\quad\ \left\Vert u_{d}^*\big(\breve{x},\lambda^*(\breve{x})\big) - u_{\lambda}^*\big(x,\lambda^*(x)\big) \right\Vert^2 \notag\\
&= \big\Vert  \big\{ \left[\lambda^*(\breve{x})\mathcal{L}_{\rho} s(\breve{x}) - \lambda^*(x)\mathcal{L}_{\rho} s(x)\right] - \big[ \mathcal{L}_{\rho}V^*_{\lambda}(\breve{x}) - \mathcal{L}_{\rho}V^*_{\lambda}(x) \big] \big\} R^{-1}\big\Vert^2\notag\\
&\leq 8\overline{\big\Vert \lambda^*(x)\mathcal{L}_{\rho}s(x) R^{-1}\big\Vert}^2 + 2 d_v^2 \Vert \breve{e} \Vert^2
\end{flalign}
Substituting \eqref{mkl}, \eqref{554} and \eqref{555} into \eqref{556}, we can obtain that
\begin{flalign}
\dot{V}_{\lambda}^*(x) \leq & - (1-\mu)\lambda_{\min}(Q)\Vert x \Vert^2 + 2d_v^2 \Vert R \Vert \Vert \breve{e} \Vert^2 -\mu \lambda_{\min}(Q)\Vert x \Vert^2 + C_1,
\end{flalign}
where $0<\mu<1$, $C_1 = \frac{17}{2} \overline{\left\Vert \lambda^* \mathcal{L}_{\rho} s R^{-\frac{1}{2}}\right\Vert}^2$. If the condition \eqref{etm} holds, it follows
\begin{flalign}
\dot{V}_{\lambda}^*(x) \leq -\mu \lambda_{\min}(Q)\Vert x \Vert^2 + C_1.
\end{flalign}
Therefore, the Lyapunov candidate $V_{\lambda}^*$ satisfies $\dot{V}_{\lambda}^* < 0$ if $\Vert x \Vert > \sqrt{ \frac{C_1}{\mu\lambda_{\min}(Q)} }$. Using \cite[Theorem 4.18]{hassan2002nonlinear}, we can conclude that the state $x$ is UUB.
\end{proof}

Subsequently, we extend the event-triggered condition \eqref{etm} to a self-triggered one. With the relationship $\breve{e}_j(t) = x(t)-\breve{x}_j(t)$ as well as inequality $2a^\top b \ge -\pi a^\top a - \frac{1}{\pi} b^\top b,\ \pi\in(0,1)$, we can derive that
\begin{flalign}\label{ppo}
\Vert x(t) \Vert^2 &= \breve{x}_j^\top\breve{x}_j + 2\breve{x}_j^\top\breve{e}_j(t) + \breve{e}_j(t)^\top\breve{e}_j(t)\notag \\
&\ge(1-\pi)\Vert \breve{x}_j \Vert^2 - (\frac{1}{\pi}-1)\Vert \breve{e}_j(t) \Vert^2. 
\end{flalign}
Substituting the inequality \eqref{ppo} into the event-triggering condition \eqref{etm} yields
\begin{flalign}
& f_{\rm v,event}(\Vert x \Vert) \ge \sqrt{ \frac{\chi_1 \Vert \breve{x}_j \Vert^2 - \chi_2\Vert \breve{e}_j \Vert^2 }{2d_v^2 \Vert R \Vert} } ,\ t\in[t_j,t_{j+1}), 
\end{flalign}
where $\chi_1=(1-\mu)(1-\pi)\lambda_{\min}(Q)$, $\chi_2=(1-\mu)(\frac{1}{\pi}-1)\lambda_{\min}(Q)$. Through simple mathematical transformation, we can obtain the self-triggering threshold as
\begin{flalign}\label{self-triggered me}
& \Vert \breve{e}_j(t) \Vert \leq \sqrt{\frac{ \chi_1\Vert \breve{x}_j \Vert^2}{2d_v^2\Vert R \Vert + \chi_2 } } \triangleq f_{\rm v,self}(\Vert\breve{x}_j\Vert),\ t\in[t_j,t_{j+1}), 
\end{flalign}
where $f_{\rm v,self}(\Vert\breve{x}_j\Vert)$ denotes the self-triggering threshold function, which can be determined by current sampled state $\breve{x}_j$.

\subsection{Triggering threshold for safety guarantee}
Intuitively, the designed RCBF safety constraints \eqref{rcbf} are only activated when the system trajectory approaches the safety boundaries, meaning we do not need to calculate the RCBF-based safeguard controller $\lambda^*R^{-1}\mathcal{L}_{\rho}s(x)^\top$ in real time. Similarly to the triggering condition \eqref{self-triggered me} for stability guarantees discussed in the previous subsection, we propose a self-triggered update approach for the RCBF-based safeguard controller. This method effectively reduces computational costs while ensuring that no safety restrictions are violated. 

Similarly, we assume that there exists an identifier for the uncertain system \eqref{system} with a Lyapunov function $V_{\theta}$ characterizing the convergence of the estimation error $\tilde{\theta}$ as \eqref{theta lyapunov}. Consider the robust safe set $\mathcal{D}_{\theta}\triangleq\{x\in\mathcal{X}:s_{\theta}(x,t)\ge0\}$ with $s_{\theta}(x,t)\triangleq s(x) - \eta V_{\theta}$. 
\begin{thm}
Suppose that $s_{\theta}(x(0),0)>0$ and $\frac{\partial s}{\partial x}\neq 0,\ \forall x\in\mathcal{D}$. At the triggering instant $t_j$, the self-triggered RCBF-based safety constraint imposed on a self-triggered state-feedback controller $u_d$ is given as
\begin{flalign}\label{nu_d}
\nu_d(\breve{x}_j,\hat{\theta},u_d)=\mathcal{L}_\omega s(\breve{x}_j)\hat{\theta}(t_j) + \mathcal{L}_{\rho} s(\breve{x}_j) u_d(\breve{x}_j) - \Xi(\breve{x}_j) + \gamma\alpha \big(s(\breve{x}_j)\big)\ge 0,
\end{flalign}
where $\gamma\in(0,1)$ is a user-defined constant. Then, the self-triggering threshold for safety gurantee is provided as
\begin{flalign}\label{self-trigger cbf}
    \Vert \breve{e}_j \Vert \leq f_{\rm s,self}(\Vert \breve{x}_j \Vert),
\end{flalign}
where $f_{\rm s,self}(\Vert \breve{x}_j \Vert)$ is defined in two cases:
\begin{flalign}
f_{\rm s,self}(\Vert \breve{x}_j \Vert) = \left\{ \begin{array}{ll}
   \overline{M}_j^{-1}\Big( \nu_d(\breve{x}_j,\hat{\theta}(t_j),u_d(\breve{x}_j)) \Big),  &\nu_d \ge (1-\gamma)\alpha_s\cdot s(\breve{x}_j) , \\
   \overline{M}_j^{-1}\Big( (1-\gamma)\alpha_s\cdot s(\breve{x}_j) \Big),  & \nu_d<(1-\gamma)\alpha_s\cdot s(\breve{x}_j),
\end{array}\right.
\end{flalign}
where $\overline{M}_j:\mathbb{R}_{\ge 0}\mapsto\mathbb{R}_{\ge 0},\ j=1,2,\dots$ is a monotonically increasing function concerning the norm of sampling error $\Vert \breve{e}_j \Vert$, the specific definition of which will be detailed in the following proof. Under this triggering threshold, we can conclude that the system trajectory is safe.
\end{thm}
\begin{proof}
If the self-triggered RCBF-based safety constraint $\nu_d$ is calculated at triggering instant $t_j$, the safe period $T_{\rm s,j}$ corresponding to the self-triggering threshold for safety gurantee $f_{\rm s,self}(\Vert \breve{x}_j \Vert)$ can be obtained based on the current sampled state $\breve{x}_j$ as \cite{ming2022self}
\begin{flalign}\label{safe period}
    T_{\rm s,j}(\Vert \breve{e}_j \Vert)= \frac{1}{l_1+l_2}\ln\left( 1+ \frac{l_1+l_2}{l_1\Vert \breve{x}_j\Vert+l_3}\Vert \breve{e}_j \Vert \right).
\end{flalign}
The definitions of the constants $l_1$, $l_2$, $l_3$, and the detailed derivation process of $T_{\rm s,j}$ can be found in \ref{appendix2}. For $t\in[t_j,t_j+T_{\rm s,j})$, we define that
\begin{flalign}
 M_j(t) =&  \mathcal{L}_\omega s(x)\theta- \mathcal{L}_\omega s(\breve{x}_j)\theta +\mathcal{L}_{\rho} s(x)u_d -\mathcal{L}_{\rho} s(\breve{x}_j)u_d\notag \\
&+(\eta -\eta_c )k_3\left( V_{\theta}(\tilde{\theta}(t)) - V_{\theta}(\tilde{\theta}(t_j)) \right) + \alpha\big(s(x)\big)-\alpha\big(s(\breve{x}_j)\big), 
\end{flalign}
where $\eta > \eta_c > 0$. Using the inequality $\dot{V}_{\theta}\leq -k_1k_3 \Vert \tilde{\theta} \Vert^2$, we have
\begin{flalign}
& V_{\theta}(\tilde{\theta}(t)) - V_{\theta}(\tilde{\theta}(t_j)) \leq -k_1k_3\Vert \tilde{\theta}(t_j) \Vert^2 T_{\rm s,j}.
\end{flalign}
Let $d_1$, $d_2$, and $d_3$ denote the Lipschitz constants of $\mathcal{L}_\omega s(\cdot)$, $\mathcal{L}_{\rho} s(\cdot)$ and $s(\cdot)$. The term $M(t)$ can be bounded as follows:
\begin{flalign}\label{M(t)}
\vert M(t) \vert  \leq &\Bigg[ d_1\overline{\Vert \theta \Vert} + d_2 \overline{\Vert u_d \Vert}+ \frac{d_3 k_3\eta_c }{\eta} \Bigg] \Vert \breve{e}_{j} \Vert \notag \\[4pt]
& + (\eta -\eta_c )k_1 k_3^2\overline{\left\Vert \tilde{\theta} \right\Vert}^2T_{\rm s,j}( \Vert\breve{e}_j\Vert ) \triangleq  \overline{M}_j(\Vert \breve{e}_j \Vert). 
\end{flalign}
Taking the time derivative of $s_{\theta}(x,t)$ yields
\begin{flalign}\label{self-triggered CBF proof}
&\dot{s}_{\theta}(x,t) = \dot{s}(x) - \eta \dot{V}_{\theta} \notag \\[3pt]
&\ge \mathcal{L}_{\omega} s(x)(\hat{\theta}+\tilde{\theta}) + \mathcal{L}_{\rho} s(x) u_d(\breve{x}_j) + \eta k_3 V_{\theta} \notag \\[3pt]
& = \mathcal{L}_\omega s(\breve{x}_j)\hat{\theta}(t_j) + \mathcal{L}_{\rho} s(\breve{x}_j)u_d(\breve{x}_j)+\mathcal{L}_\omega s(\breve{x}_j)\tilde{\theta}(t_j) \notag \\[3pt]
&\quad + (\eta -\eta_c )k_3 V_{\theta}(\tilde{\theta}(t_j),t_j) + \eta_c k_3 V_{\theta}(\tilde{\theta}(t),t) + \alpha\big(s(x)\big)-\alpha\big(s(\breve{x}_j)\big) + M(t) \notag \\[3pt]
&\ge \mathcal{L}_\omega s(\breve{x}_j)\hat{\theta}(t_j)+\mathcal{L}_{\rho} s(\breve{x}_j)u_d(\breve{x}_j) - \Xi(\breve{x}_j) + \alpha\big(s(x)\big)-\alpha\big(s(\breve{x}_j)\big) + \eta_c k_3 V_{\theta} \notag\\[3pt]
&\quad + \left( \sqrt{\varpi}\Vert \mathcal{L}_\omega s(\breve{x}_j) \Vert - \frac{ \Vert \tilde{\theta}(t_j) \Vert }{2\sqrt{\varpi}}  \right)^2  + M(t)\notag\\[3pt]
&\ge -\frac{\eta_c k_3}{\eta}\cdot s_{\theta} + \nu_d(\breve{x}_j,\hat{\theta}(t_j),u_d(\breve{x}_j)) + (1-\gamma)\alpha\big( s(\breve{x}_j\big) -\overline{M}(\Vert\breve{e}_j\Vert).
\end{flalign}
Combining the self-triggered mechanism \eqref{self-trigger cbf}, we can derive that $\dot{s}_{\theta}\ge -\frac{\eta_c k_3}{\eta} s_{\theta}$, which verifies the safety guarantee during the triggering interval.
\end{proof}
\begin{remark}
We design the self-triggered mechanism for the calculation of self-triggered RCBF \eqref{nu_d} in two cases. The first case occurs when the system trajectory has not yet approached the safety constraints or remains sufficiently far from them, i.e., when a safety margin still exists in the triggering instant $t_j$: \( \nu_d\big(\breve{x}_j,\hat{\theta}(t_j),u_d(\breve{x}_j) \big) > 0\). The other scenario occurs when the RCBF constraint becomes active and the safety margin diminishes to zero, i.e., $\nu_d\big(\breve{x}_j,\hat{\theta}(t_j),u_d(\breve{x}_j) \big) = 0$. In this case, Zeno behavior may arise, where the triggering interval between consecutive RCBF computations approaches zero, a phenomenon that must be avoided.
\end{remark}

\subsection{Self-triggered RCBF-constrained optimal controller}
Combining the triggering thresholds for stability and safety guarantees \eqref{self-triggered me},\eqref{self-trigger cbf}, the self-triggered mechanism can be designed as
\begin{flalign}\label{self}
    \Vert \breve{e}_j \Vert \leq \min\{ f_{\rm v,self}(\Vert\breve{x}_j\Vert), f_{\rm s,self}(\Vert\breve{x}_j\Vert) \}
\end{flalign}
Similar to the construction of $T_{\rm s,j}$, we can derive the stability period $T_{\rm v,j}$ corresponding to the triggering threshold $f_{\rm v,self}(\Vert \breve{e}_j \Vert)$. Hence, the triggering interval $T_j$ can be determined as $T_j=\min\{T_{\rm v,j}, T_{\rm s,j}\}$. This ensures that the minimum inter-event time $\Delta t_{\min} = \min_{j\in\mathbb{N}}\{ T_j \}$ remains strictly positive, thereby preventing Zeno behavior.

Compared to RCBF \eqref{rcbf} for time-triggered control, self-triggered RCBF \eqref{nu_d} introduces only an additional adjustment factor \(\gamma\). We can define the optimal SVF $V_{\lambda}^*(x)$ with the optimal Lagrange multiplier function $\lambda_d^*:\mathbb{R}^n\mapsto\mathbb{R}_{\ge 0}$ induced by the self-triggered RCBF constraint \eqref{nu_d} as
\begin{flalign}
    V_{\lambda}^*(x) = \min_{u\in\mathcal{U}}\max_{\lambda_d\in\mathbb{R}_{\ge0}}\int_{t}^{\infty} \Big( l(x,u) - \lambda_d(x)\nu_d(x,\theta,u) \Big){\rm d}\tau.
\end{flalign}
Similarly, we can obtain the HJB equation augmented with the self-triggered RCBF \eqref{nu_d} as: for triggering instants $t=t_j,j=1,\dots,$
\begin{flalign}\label{self-triggered HJB}
l\left(x,u_d^*\right) -\lambda_d^*(x)\nu_d(x,\theta,u_d^*) + \nabla V_{\lambda}^*(x)\big[ \zeta(x,\theta) + \rho(x) u_d^* \big]= 0.  
\end{flalign}
The corresponding optimal safe controller $u_{\lambda}^*\big(x,\lambda_d^*(x)\big)$ takes the form of \eqref{optimal safe controller} and the optimal Lagrange multiplier $\lambda_d^*(x)$ can be derived as 
\begin{flalign}\label{lamda-d}
    \lambda_d^*(x) = \max\left( -\frac{\nu_d\big(x,\theta,u_{\rm no}(x)\big)}{R_{s\rho}(x)}, 0  \right)
\end{flalign}
with $u_{\rm no}(x) = -R^{-1}\rho(x)^\top \nabla V_{\lambda}^{*}(x)^\top$.

Due to the system's nonlinearity and the presence of uncertain parameters in the model, it is challenging to obtain an analytical solution for the self-triggered constrained HJB equation \eqref{self-triggered HJB}. In the next section, we will propose an online safety-embedded critic learning framework to approximate the optimal SVF $(V_{\lambda}^*)$ and the corresponding optimal safe controller $(u_{\lambda}^*)$.

\section{Online safety-embedded critic learning framework}\label{sec4}
In this section, we address the third challenge posed by the aforementioned control objectives: to numerically solve the constrained HJB equation \eqref{constrained HJB}, we develop a three-component online learning architecture that combines: parameter identification, value function approximation, and safety constraint enforcement. The parameter estimation module was previously detailed. We now concentrate on the adaptive critic learning framework with integrated safety constraints via adaptive Lagrange multiplier \cite{bandyopadhyay2023hjb}. 

\subsection{Parameter estimation}
We maintain parameter estimates $\hat{\theta}$ that generate the approximate drift dynamic $\hat{\zeta} = (x,\hat{\theta})\triangleq \omega(x)\hat{\theta}$. The identification of uncertain parameters adopts the method proposed in \cite{mahmud2021safe}. We denote the input coupling as $\varrho(x,u) = \rho(x)u$ and define the intermediate integral variables as follows:
\begin{flalign}
    & \dot{\Omega} = \left\{ \begin{array}{cl}
        \omega(x), & \Vert \Omega_f \Vert \leq \overline{\Omega}_f, \\
         0, & {\rm otherwise}, 
    \end{array} \right. \quad \Omega(0)=0, \\
    & \dot{\Omega}_f = \left\{ \begin{array}{cl}
        \Omega^\top\Omega, & \Vert \Omega_f \Vert \leq \overline{\Omega}_f, \\
         0, & {\rm otherwise}, 
    \end{array} \right. \quad \Omega_f(0)=0, \\
    & \dot{\varrho}_f = \left\{ \begin{array}{cl}
        \varrho(x,u), & \Vert \Omega_f \Vert \leq \overline{\Omega}_f, \\
         0, & {\rm otherwise}, 
    \end{array} \right. \quad \varrho_f(0)=0, \\
    & \dot{\Psi}_f = \left\{ \begin{array}{cl}
        \Omega^\top\big( x(t) - x(0) - \varrho_f \big), & \Vert \Omega_f \Vert \leq \overline{\Omega}_f, \\
         0, & {\rm otherwise}, 
    \end{array} \right. \quad \Psi_f(0)=0,
\end{flalign}
where $\overline{\Omega}_f$ is a user-defined upper bound for the immediate integral variable $\Omega_f$. The update law for the estimated parameters is given by:
\begin{flalign}\label{adaptation law for parameters}
    \dot{\hat{\theta}} = \Gamma_{\theta}\Omega_f^\top\big( x(t) - x(0) - \varrho_f \big),\quad \hat{\theta}(0) = \theta_0,
\end{flalign}
where $\Gamma_{\theta}>0$ is a constant adaptation gain matrix.

\begin{assumption}\label{pe theta}
(Finite excitation condition, \cite{mahmud2021safe}) There exist a time instant $T_{\theta}>0$ and a strictly positive constant $a>0$ such that $\lambda_{\min}\left(\Omega_f^\top(t)\Omega_f(t)\right) > a >0$ for all $t>T_{\theta}$. 
\end{assumption}
Define the estimated error vector as $\tilde{\theta}\triangleq \theta - \hat{\theta}$. Let $V_{\theta} = \frac{1}{2}\tilde{\theta}^\top \Gamma_{\theta}^{-1}\tilde{\theta}$ be a Lyapunov candidate function. According to \cite{mahmud2021safe}, suppose that Assumption \ref{pe theta} holds, then update law \eqref{adaptation law for parameters} ensures convergence of the parameter estimation error, with the conclusion $\dot{V}_{\theta} \leq -a\left\Vert \tilde{\theta} \right\Vert^2$.

\begin{remark}
Note that the regression vector $\Omega_f$ grows over time. \cite{mahmud2021safe} employs a bounded update method, halting updates once $\Vert \Omega_f \Vert$ reaches a preset limit $\overline{\Omega}_f$, which may trap $\hat{\theta}$ in a local optimum. We introduce a refresh mechanism: if $\Vert \Omega_f \Vert$ hits the bound at $t_{\rm re}$, integration resets from $t_{\rm re}$:
\begin{flalign}
& \dot{\Omega}_{\rm re,f} = \left\{ \begin{array}{cl}
        \Omega_{\rm re}^\top(t)\Omega_{\rm re}(t), & \Vert \Omega_{\rm re,f} \Vert \leq \overline{\Omega}_f, \\
         0, & {\rm otherwise}, 
    \end{array} \right. \quad \Omega_{\rm re,f}(t_{\rm re})=0, \\
& \dot{\Psi}_{\rm re,f} = \left\{ \begin{array}{cl}
       \Omega_{\rm re}^\top(t)\big( x(t) - x(t_{\rm re}) - \varrho_{\rm re}(t) \big), & \Vert \Omega_{\rm re,f} \Vert \leq \overline{\Omega}_f, \\
         0, & {\rm otherwise}, 
    \end{array} \right. \quad \Psi_{\rm re,f}(0)=0,
\end{flalign}
where $\Omega_{\rm re}(t)\triangleq\Omega(t) - \Omega(t_{\rm re})$ and $\varrho_{\rm re}(t)\triangleq \varrho(x(t),u(t)) - \varrho(x(t_{\rm re}),u(t_{\rm re}))$. This technique will be demonstrated in the numerical validation section.
\end{remark}

\subsection{Value function approximation}
We employ StaF kernels \cite{kamalapurkar2016efficient,deptula2021approximate} to approximate the optimal SVF $V_{\lambda}^*$ corresponding to the self-triggered HJB equation \eqref{self-triggered HJB}. The crux of using StaF kernels lies in the fact that we only need to perform a local approximation of the value function around the current state $x$ to achieve satisfactory results. Moreover, as a local approximation method, it significantly reduces the number of basis functions required and enables efficient online learning with reduced computational overhead. Specifically, the local approximation kernel functions are defined with state-following centers, $\upsilon:\mathcal{X}\to\mathcal{X}$, which satisfy $\upsilon(x)\in B_r(x)$. Hence, the optimal SVF associated with \eqref{self-triggered HJB} can be described as
\begin{flalign}\label{st}
& V_{\lambda}^*(x_l) = W_s(x)^\top \sigma(x_l,\upsilon(x))+\epsilon_{\upsilon}(x,x_l),\quad x_l\in \mathcal{B}_r(x), 
\end{flalign}
where $W_s : \mathcal{X}\to\mathbb{R}^L$ is the local ideal weight function, $\sigma : \mathcal{X}\times\mathcal{X}^L\to\mathbb{R}^L$ is the vector of local approximation kernels, and $\epsilon_{\upsilon} : \mathcal{X}\times\mathcal{X}\to\mathbb{R}$ is the corresponding approximation error. 

Given that the local ideal weight $W_s(x)$ is unknown in \eqref{st}, we can represent an approximate safety-embedded value function as $\hat{V}_{\lambda}$:
\begin{flalign}\label{es}
&\hat{V}_{\lambda}(x_l,x,\hat{W}_s) = \hat{W}_s^\top \sigma(x_l,\upsilon(x)), 
\end{flalign}
where $\hat{W}_s\in\mathbb{R}^L$ denotes the estimated critic weight vector. Further, we can obtain the corresponding approximate optimal safe control law as
\begin{flalign}\label{es-u}
\hat{u}_{\lambda}(x,\hat{\theta},\hat{W}_s) = R^{-1}\left[ \hat{\lambda}_d\left(x,\hat{\theta},\hat{W}_s\right)\mathcal{L}_{\rho} s(x)^\top - \rho(x)^\top \nabla \sigma(x,\upsilon(x))^\top \hat{W}_s \right], 
\end{flalign}
where $\hat{\lambda}_d\left(x,\hat{\theta},\hat{W}_s\right)$ is an estimate of the Lagrange multiplier in \eqref{lamda-d} and can be described as
\begin{flalign}\label{es-lam}
& \hat{\lambda}_d\left(x,\hat{\theta},\hat{W}_s\right)=\max\left( -\frac{\nu_d\left(x,\hat{\theta},\hat{u}_{\rm no}(x)\right)} {R_{s\rho}(x)}, 0\right), 
\end{flalign}
with $\hat{u}_{\rm no}(x) = -R^{-1}\rho(x)^\top \sigma(x,\upsilon(x))^\top \hat{W}_s$.
\begin{remark}
Note that under the self-triggered mechanism, $\hat{\lambda}_d\left(\breve{x}_j,\hat{\theta},\hat{W}_s\right)$ is computed based on sampled states $\breve{x}_j$ and maintained via a zero-order holder (ZOH). The corresponding controller under ZOH maintenance can be denoted as $\hat{u}_d = \left\{ \begin{array}{ll}
    \hat{u}_{\lambda}\left(x(t_j),\hat{\theta}(t_j),\hat{W}_s(t_j)\right), & t\in[t_j,t_{j+1}) \\
    \hat{u}_{\lambda}\left(x(t_{j+1}),\hat{\theta}(t_{j+1}),\hat{W}_s(t_{j+1})\right), & t = t_{j+1}
\end{array} \right. .$
\end{remark}
For $t\in[t_j,t_{j+1})$, substituting the estimated vectors $\hat{\theta}$, $\hat{W}_s$, and $\hat{u}_d(t_j)$ into the constrained HJB equation \eqref{self-triggered HJB} results in the Bellman error (BE)
\begin{flalign}\label{be}
\delta_t\left(x,\hat{\theta},\hat{W}_s,\hat{u}_d\right) = &\nabla \hat{V}_{\lambda}(x,\hat{W}_s)\left(\hat{\zeta}\left(x,\hat{\theta}\right)+\rho(x)\hat{u}_{d}(t_j) \right) + l\left(x,\hat{u}_{d}(t_j)\right). 
\end{flalign}
Using the ``Bellman extrapolation" \cite{kamalapurkar2016efficient} and ``experience replay" \cite{yang2019adaptive} techniques, we evaluate the BE at a set of historical sampled states $\{x_{ti}\in\mathbb{R}^n:x_{ti}=x(t_i),t_i\in[t_j,t),i=1,\dots,N \}$, as
\begin{flalign}
\delta_{ti}\left(x_{ti},\hat{W}_s,\hat{\theta},\hat{u}_d\right)\triangleq &\nabla \hat{V}_{\lambda}\left(x_{ti},\hat{W}_s\right)\left(\hat{\zeta}\left(x_{ ti},\hat{\theta}\right)+\rho(x_{ti})\hat{u}_d(t_j) \right) + l\left(x_{ ti},\hat{u}_d(t_j)\right). 
\end{flalign}

\begin{remark}
Note that the approximate optimal safe control policy, as given by \eqref{es-u}, satisfies the self-triggered RCBF-based safety constraint \eqref{nu_d}, thereby ensuring system safety. Additionally, we can conclude that $\hat{\lambda}_d\left(x,\hat{\theta},\hat{W}_s\right)\cdot\nu_d\left(x,\hat{\theta},\hat{u}_d\right)=0$, and as a result, we disregard this term in the BE error, as shown in \eqref{be}. 
\end{remark} 

In line with the recursive least-squares-based method, the estimated critic weight vector $\hat{W}_s$ is updated to minimize the accumulated error $E(t) = \int_{t_0}^t \left( \delta_t^2(\tau) + \sum_{i=1}^N\delta_{ti}^2(\tau) \right)d\tau $, as
\begin{subequations}
\begin{align}
&\dot{\hat{W}}_s(t) = - k_{c1}\Gamma(t)\frac{\xi(t)}{\iota(t)}\delta_t -\frac{k_{c2}}{N}\Gamma(t)\sum_{i=1}^N \frac{\xi_i(t)}{\iota_i(t)}\delta_{ti}, \label{critic law}\\[3pt]
&\dot{\Gamma}(t) = \beta\Gamma(t) - k_{c1}\Gamma(t)\frac{\xi(t)\xi(t)^\top}{\iota(t)^2}\Gamma(t) - \frac{k_{c2}}{N}\Gamma(t)\sum_{i=1}^N \frac{\xi_i(t)\xi_i(t)^\top}{\iota_i(t)^2}\Gamma(t), \label{Gamma law}
\end{align}
\end{subequations}
where $\xi(t)=\nabla\sigma\left( \hat{\zeta} + \rho\hat{u}_d \right)$ is the regressor vector, $\iota(t):=\sqrt{1+\gamma \xi(t)^\top \xi(t)}$ is the normalized term, $k_{c1}, k_{c2}>0$ and $\beta>0$ are user-defined learning rate and forgetting constant, respectively.

\subsection{Stability analysis}
To guarantee the convergence of the true states $x$, the estimated parameters $\hat{\theta}$ and critic weights $\hat{W}_s$, the following excitation conditions, boundedness assumptions and Lipschitz continuity of the optimal controller are established.
\begin{assumption}\label{as2}
\cite{kamalapurkar2016efficient} There exists a strictly positive constant $b\in\mathbb{R}_{> 0}$ such that for all $t\ge 0$, 
\begin{flalign}
0< b \leq \inf_{t\ge 0} \left\{\lambda_{\min}\left( \frac{k_{c1}\xi(t)\xi(t)^\top}{\iota(t)^2}+\frac{k_{c2}}{N} \sum_{i=1}^N\frac{\xi_i(t)\xi_i(t)^\top}{\iota_i(t)^2} \right)\right\} 
\end{flalign}
\end{assumption}
\begin{remark}
This is the relaxed rank-like persistently exciting (PE) condition induced by the concurrent learning technique. It can be observed that by incorporating additional sampled state trajectories for BE extrapolation, the above rank-like PE condition can be more readily satisfied. Moreover, under this relaxed PE condition, the least-squares matrix gain $\Gamma$ is bounded as $\underline{\Gamma} \mathbb{I}_L \leq \Gamma(t) \leq \overline{\Gamma} \mathbb{I}_L$, where $\overline{\Gamma} > \underline{\Gamma}>0$.
\end{remark}

\begin{thm}
Suppose that the matrix $B = \left[\begin{array}{ccc}
        \mu\lambda_{\min}(Q)\mathbb{I}_n  & \textbf{0}_{n\times l}  & \textbf{0}_{n\times p} \\
        \textbf{0}_{l\times n} & \vartheta_c & \vartheta_{c\theta} \\
        \textbf{0}_{p\times n} & \vartheta_{c\theta}^T  & a\mathbb{I}_p 
\end{array}\right]$
with $\vartheta_c =\left( \frac{b}{2} + \frac{\beta}{2\overline{\Gamma}} \right) + 4\left\Vert R^{-1}\right\Vert\overline{\left\Vert \tilde{\lambda}\mathcal{L}_{\rho}s \right\Vert}^2$ and $\vartheta_{c\theta} = k_{c1}\frac{\xi}{2\iota}W_s^\top\nabla\sigma \omega + \frac{k_{c2}}{2N}\sum_{i=1}^N\frac{\xi_i}{\iota_i}W_s^\top\nabla \sigma_i \omega_i$ is strictly positive definite and Assumption \ref{pe theta}, \ref{as2} hold. Under the self-triggered mechanism \eqref{self}, online estimation methods for parametric uncertainties \eqref{adaptation law for parameters}, critic weights \eqref{critic law}-\eqref{Gamma law} and Lagrange multipliers \eqref{es-lam}, the estimation error of critic weights ($\tilde{W}_s$) of the approximate optimal SVF \eqref{es}, and the system state ($x$) driven by the approxiamte optimal safe control policy \eqref{es-u} are uniformly ultimately bounded (UUB).
\end{thm}
\begin{proof}
Let $Z = (x^\top, \breve{x}^\top,\tilde{W}_s^\top,\tilde{\theta}^\top)^\top$ be the extended state. Consider the following Lyapunov function candidate
\begin{flalign}
& \mathcal{V}(Z) = V_{\lambda}^*(x) + V_{\lambda}^*(\breve{x}_j) + V_c(\tilde{W}_s) + V_{\theta}(\tilde{\theta}), 
\end{flalign}
where $V_{\lambda}^*(x)$ and $V_{\lambda}^*(\breve{x}_j)$ are optimal SVF for the time-triggered and self-triggered systems, $V_c(\tilde{W}_s) = \frac{1}{2}\tilde{W}_s^\top \Gamma^{-1}(t)\tilde{W}_s$, and $V_{\theta}(\tilde{\theta})=\frac{1}{2}\tilde{\theta}^\top \Gamma_{\theta}^{-1}(t)\tilde{\theta}$. 

Case \uppercase\expandafter{\romannumeral1}: Let $t_j\leq t < t_{j+1}$. Note that during the triggering interval, $\dot{V}_{\lambda}^*(\breve{x}_j)=0$. Taking time derivative of $V_{\lambda_d}^*(x)$ along the solution of 
\begin{flalign}
\dot{V}_{\lambda}^* &= \nabla V_{\lambda}^*\left(\zeta(x,\theta)+\rho(x)\hat{u}_{d}(t_j) \right) \notag\\
&= \nabla V_{\lambda}^*\left(\zeta(x,\theta)+\rho(x) u_{\lambda}^*\left(x,\lambda_d^*(x)\right) \right) + \nabla V_{\lambda}^* \rho \left( \hat{u}_{d}(t_j)-u_{\lambda}^*\left(x,\lambda_d^*(x)\right) \right) \notag \\
&= -x^TQx - \frac{1}{2}u_{\lambda}^{*\top}Ru_{\lambda}^* +\nabla V_{\lambda}^* \rho \left( \hat{u}_{d}(t_j)-u_{\lambda}^*\left(x,\lambda_d^*(x)\right) \right). \label{v-self}  
\end{flalign}
Note that $\nabla V_{\lambda}^*\rho = \lambda_d^*\mathcal{L}_{\rho} s - u_{\lambda}^{*\top} R$. We can obtain that
\begin{flalign}
\nabla V_{\lambda}^*\rho( \hat{u}_{d} &- u_{\lambda}^*) =\underbrace{\lambda_d^* \mathcal{L}_{\rho} s\big( (\hat{u}_d - u_d^*) + (u_d^* - u_{\lambda}^*) \big)}_{\Upsilon_3} \notag\\
&\underbrace{ -u_{\lambda}^{*\top}R\big( (\hat{u}_d - u_d^*) + (u_d^* - u_{\lambda}^*) \big) }_{\Upsilon_4}. 
\end{flalign}
Using Young's inequality, we have that
\begin{flalign}
& \Upsilon_3 \leq \frac{1}{2}\left\Vert \lambda_d^* \mathcal{L}_{\rho} s R^{-\frac{1}{2}}\right\Vert^2 + \frac{1}{2}\left\Vert R \right\Vert \left\Vert (\hat{u}_d - u_d^*) + (u_d^* - u_{\lambda}^*) \right\Vert^2, \label{up3}\\
& \Upsilon_4 \leq \frac{1}{2}u_{\lambda}^{*\top}Ru_{\lambda}^* + \frac{1}{2}\Vert R \Vert \left\Vert (\hat{u}_d - u_d^*) + (u_d^* - u_{\lambda}^*) \right\Vert^2. \label{up4}
\end{flalign}
According to \eqref{mkl}, we have
\begin{flalign}\label{u-ul}
&\left\Vert u_{d}^*\big(\breve{x},\lambda_d^*(\breve{x})\big) - u_{\lambda}^*\big(x,\lambda_d^*(x)\big) \right\Vert^2 \leq 8\overline{\big\Vert \lambda_d^*(x)\mathcal{L}_{\rho}s(x) R^{-1}\big\Vert}^2 + 2 d_v^2 \Vert \breve{e} \Vert^2
\end{flalign}
In addition, we have
\begin{flalign}\label{ud-u}
&\Vert \hat{u}_d-u_d^* \Vert^2 = \Vert R^{-1}\rho^\top\nabla\sigma^\top\tilde{W}_s - R^{-1}\tilde{\lambda}_d \mathcal{L}_{\rho} s^\top - \epsilon_u \Vert^2 \notag\\
& \leq 2 \Vert R^{-1} \Vert^2 \overline{\Vert \rho\nabla\sigma \Vert}^2 \Vert \tilde{W}_s \Vert^2 + 2 \Vert R^{-1} \Vert^2 \overline{\Vert \tilde{\lambda}_d \mathcal{L}_{\rho} s \Vert}^2 + 2\overline{\Vert \epsilon_u \Vert}^2 
\end{flalign}

Substituting \eqref{up3}-\eqref{ud-u} into \eqref{v-self} yields
\begin{flalign}
& \dot{V}_{\lambda}^*(x)\leq -\mu\lambda_{\min}(Q)\Vert x \Vert^2 + 4\Vert R^{-1}\Vert \overline{\Vert \tilde{\lambda}_d \mathcal{L}_{\rho} s \Vert}^2 \Vert \tilde{W}_s \Vert^2 + C_2, 
\end{flalign}
where $C_2 = 4\Vert R^{-1} \Vert \overline{ \Vert \tilde{\lambda}_d\mathcal{L}_{\rho} s \Vert }^2 + \frac{33}{2}\overline{\Vert \lambda_d^* \mathcal{L}_{\rho} s R^{-\frac{1}{2}}\Vert}^2 + 4 \Vert R \Vert  \overline{\Vert \epsilon_u \Vert}^2$.

The BE in \eqref{be} can be further rewritten in terms of the weight estimation error vector $\tilde{W}_s = W - \hat{W}_d$, $\tilde{\theta} = \theta - \hat{\theta}$, and Lagrange multiplier $\hat{\lambda}$ as
\begin{flalign}
\delta_t =& -\xi^\top\tilde{W}_s - W_s^\top\nabla \sigma \omega \tilde{\theta} + \hat{\lambda}_d\mathcal{L}_{\rho} s R^{-1} \rho^\top \nabla \sigma^\top \tilde{W}_s \notag\\
& + \hat{\lambda}_d^2 \mathcal{L}_{\rho} s R^{-1}\mathcal{L}_{\rho} s^\top + \Delta_{\epsilon}(x), 
\end{flalign}
where the function $\Delta_{\epsilon} : \mathcal{X}\to \mathbb{R}$ is uniformly bounded over $\mathcal{X}$ such that the bound $\overline{\Vert \Delta_{\epsilon} \Vert}$ decreases with decreasing $\overline{\Vert \breve{e}_j \Vert}$, $\overline{\Vert \epsilon_u \Vert}$ and $\overline{\Vert \nabla W_s \Vert}$. Note that $\hat{\lambda}$ only plays a role when the state is close to the safe set, and as the state gradually converges to the equilibrium point, $\hat{\lambda}$ also converges to zero. Therefore, the terms related to $\hat{\lambda}_d$ can be regarded as bounded small quantities. Then, we can rewrite the BE as
\begin{flalign}
\delta_t =& -\xi^\top\tilde{W}_s - W_s^\top\nabla \sigma \omega \tilde{\theta} + \Delta(x), 
\end{flalign}
Then, we take the time derivative of the second term in the Lyapunov candidate $V_c(\tilde{W}_s)$ as
\begin{flalign}
\dot{V}_c(\tilde{W}_s) &= \tilde{W}_s^\top \Gamma^{-1}(\dot{W}_s-\dot{\hat{W}_s}) - \frac{1}{2}\tilde{W}_s^\top \Gamma^{-1}\dot{\Gamma}\Gamma^{-1} \tilde{W}_s \notag \\ 
& = \tilde{W}_s^\top \Gamma^{-1} \nabla W_s(x) \left( \zeta(x,\theta) + \rho(x)\hat{u}_d \right) \notag \\ 
&\quad -\tilde{W}_s^\top \Gamma^{-1} \Big( - k_{c1}\Gamma\frac{\xi}{\iota}(-\xi^\top\tilde{W}_s-W_s^\top\nabla \sigma \omega \tilde{\theta} + \Delta) \notag\\
&\qquad\quad-\frac{k_{c2}}{N}\Gamma\sum_{i=1}^N \frac{\xi_i}{\iota_i}(-\xi_i^\top\tilde{W}_s-W_s^\top\nabla \sigma_i \omega_i \tilde{\theta} + \Delta_i) \Big) \notag\\
&\quad - \frac{1}{2}\tilde{W}_s^\top \Gamma^{-1} \Big( \beta\Gamma - k_{c1}\Gamma\frac{\xi\xi^\top}{\iota^2}\Gamma - \frac{k_{c2}}{N}\Gamma\sum_{i=1}^N \frac{\xi_i\xi_i^\top}{\iota_i^2}\Gamma \Big) \Gamma^{-1}\tilde{W}_s. 
\end{flalign}
Using the fact $\iota < \iota^2$, we can further simplify the above equation as
\begin{flalign}
\dot{V}_c \leq &\tilde{W}_s^\top \left( -\frac{\beta}{2}\Gamma^{-1}-\frac{k_{c1}}{2}\frac{\xi\xi^\top}{\iota^2}  - \frac{k_{c2}}{2N}\sum_{i=1}^N \frac{\xi_i\xi_i^T}{\iota_i^2}\right) \tilde{W}_s \notag \\
& - \tilde{W}_s^\top \left( k_{c1}\frac{\xi}{\iota}W_s^\top \nabla \sigma\omega + \frac{k_{c2}}{N}\sum_{i=1}^N \frac{\xi_i}{\iota_i}W_s^\top \nabla\sigma_i \omega_i \right) \tilde{\theta} \notag \\
& + \tilde{W}_s^\top \Gamma^{-1} \nabla W_s(x) \left( \zeta(x,\theta) + \rho(x)\hat{u}_d \right) + C_3,
\end{flalign}
where $C_3 = (k_{c1}+k_{c2})\overline{\left\Vert \frac{\xi}{\iota} \right\Vert}\overline{\left\Vert \tilde{W}_s\right\Vert}\overline{\left\Vert \Delta \right\Vert}$.

For the Lyapunov function $V_{\theta}(\tilde{\theta})$ of parameter estimation error vector, we have $\dot{V}_{\theta}\leq -a\left\Vert \tilde{\theta} \right\Vert^2$. Define the vector $\varsigma = \Big( 0,\ \varsigma_2,\ 0 \Big)^\top$, with $\varsigma_2 = \Gamma^{-1} \nabla W_s(x) \left( \zeta(x,\theta) + \rho(x)\hat{u}_d \right)$. Let $z = \left(x^\top, \tilde{W}_s^\top, \tilde{\theta}^\top \right)^\top$. Recall the definition of matrix $B$. Along with the above definition, we can derive that
\begin{flalign}
\dot{\mathcal{V}}(Z) &\leq - z^\top B z + \varsigma^\top\cdot z + C_2 + C_3 \notag \\
&\leq-(\sqrt{\lambda_{\min}(B)}\Vert z \Vert - \frac{\Vert \varsigma \Vert}{2\sqrt{\lambda_{\min}(B)}})^2 + C_2+C_3, 
\end{flalign}
Hence, the time derivative of the Lyapunov candidate $\mathcal{V}$ is negative if $\Vert z \Vert > \sqrt{\frac{C_2+C_3}{\lambda_{\min}(B)}} + \frac{\Vert \varsigma \Vert}{2\lambda_{\min}(B)}$. Using \cite[Theorem 4.18]{hassan2002nonlinear}, we can conclude that the $x$, $\tilde{W}_s$, and $\tilde{\theta}$ are UUB under the approximate self-triggered optimal safe controller $\hat{u}_d\left(x,\hat{W}_s,\hat{\theta}\right)$.

Case \uppercase\expandafter{\romannumeral2}: Let $t = t_{j+1},\ j\in\mathbb{N}$. Then we consider the difference of Lyapunov function candidate, i.e., $\Delta \mathcal{V}=V_{\lambda}^*(\breve{x}_{j+1}) - V_{\lambda}^*(\breve{x}_j)$. As proved in Case \uppercase\expandafter{\romannumeral1}, based on the conclusion that $x(t)$ is UUB, we can derive $V_{\lambda}^*(\breve{x}_{j+1})\leq V_{\lambda}^*(\breve{x}_j)$. Thus, we have $\Delta\mathcal{V}<0$.
\end{proof}

\section{Numerical validation studies}
This section presents comprehensive simulation results validating the proposed adaptive safety-embedded control framework under both time-triggered and self-triggered mechanisms. We consider the nonlinear dynamical system from \cite{jankovic2018robust}, given by $\dot{x}=\omega(x)\theta + \rho(x)u$ with
\begin{flalign}\label{sampled-data}
& \omega(x) = \left[ \begin{array}{ccc}
     x_1 & x_2 & 0 \\
     0 & 0 & x_1^3
\end{array} \right],\quad 
\rho(x) = \left[ \begin{array}{c}
     0  \\
     x_2 
\end{array} \right].
\end{flalign}
The true values of the uncertain parameter vector are $\theta_1 = -0.6,\theta_2 = -1,\theta_3=1$. The evaluation focuses on the the dual objectives of: (i) optimizing the infinite-horizon performance metric \eqref{value function} with $Q = \mathbb{I}_2,R = 1$, and (ii) ensuring asymptotic stability while maintaining strict safety guarantees.

For the approximation of SVF, we employ three kernel functions with adaptive centers, i.e., $\sigma(x,\upsilon(x))=[\sigma_1(x,\upsilon_1), \sigma_2(x,\upsilon_2), \sigma_3(x,\upsilon_3)]^\top$, where each kernel is selected as $\sigma_i(x,\upsilon_i) = e^{x^\top \upsilon_i} - 1,i=1,2,3$. The state-following centers are given by $\upsilon_i = x + b_i(x),i=1,2,3$, where $b_1(x) = 0.7\phi(x)\cdot[0,1]^\top$, $b_2(x) = 0.7\phi(x)\cdot[0.85,-0.6]^\top$, and $b_3(x) = 0.7\phi(x)\cdot[-0.85,-0.6]^\top$, and $\phi(x) = \frac{x^\top x+ 0.01}{1+ x^\top x}$. To maintain persistent excitation, we randomly sample three historical states to calculate the BE.

\subsection{Obstacle-like avoidance}
We apply the time-triggered adaptive safety-embedded optimal control policy to the uncertain system \eqref{sampled-data} with a user-defined obstacle-like safe set $\mathcal{D}=\{x:s(x)=(x_1 + 0.5)^2 + (x_2 + 1.5)^2 - 1.0^2\ge 0\}$. The learning rates and initial conditions of the identifier are set as $\Gamma_{\theta} = 100\mathbb{I}_3$ and $\hat{\theta}(0) = [0,0,0]^\top$, respectively. The upper bound for the intermediate integral variable $\Omega_f$ is set as $\overline{\Omega}_f=20$. The system is initialized as $x(0) = [-2,-3]^\top,\hat{W}_s(0)=0.1\times \textbf{1}_{3\times 1},\Gamma(0) = 10\mathbb{I}_3$, and the learning rates are selected as $k_{c1}=0.001,k_{c2}=0.001,\beta=0.002$. The RCBF \eqref{rcbf} is set with $\alpha_s = 8$ and $\Xi(x)=\frac{\Vert \mathcal{L}_{\omega}s(x) \Vert^2}{5}$.

\begin{figure}[ht]
\centering
\includegraphics[
    width=0.9\textwidth,    
    height=5cm,      
    keepaspectratio=false 
  ]{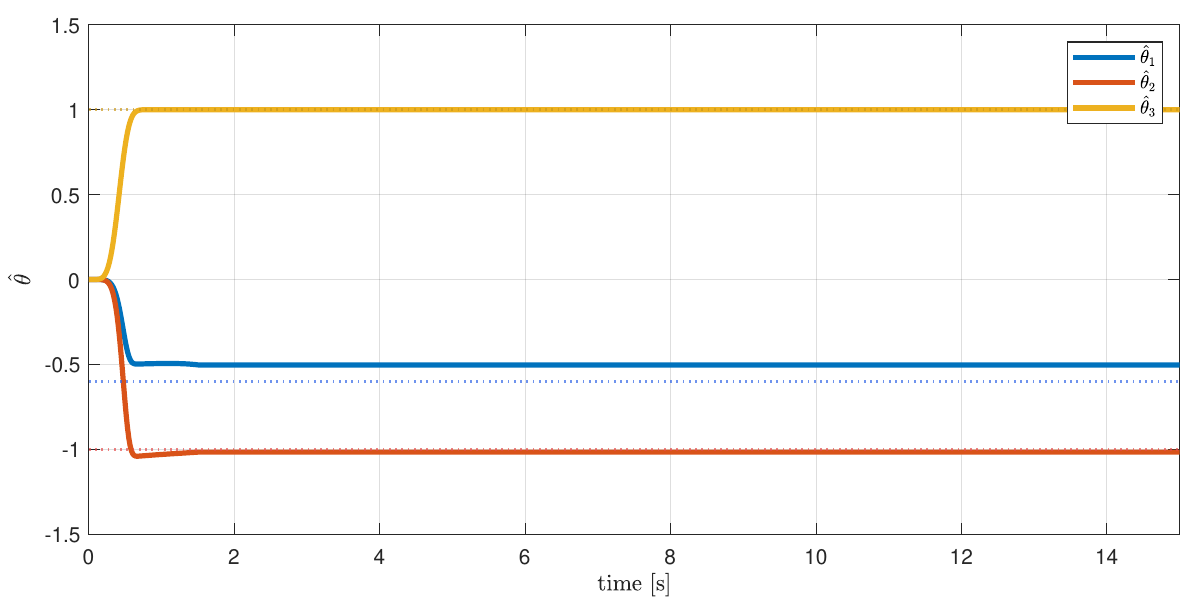}
\caption{Convergence of the estimated parameters $\hat{\theta}$}
\label{fig1}
\end{figure}

\begin{figure}[ht]
\centering
\includegraphics[
    width=0.9\textwidth,    
    height=5cm,      
    keepaspectratio=false 
  ]{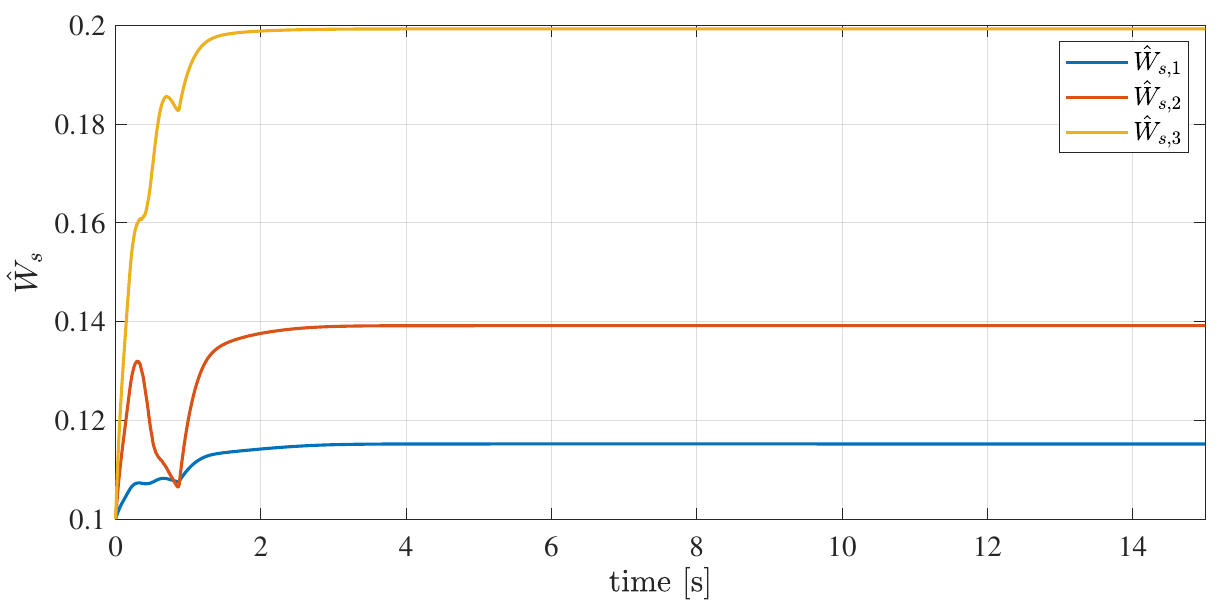}
\caption{Evolution of the approximate critic weights $\hat{W}_s$}
\label{fig2}
\end{figure}

\begin{figure}[ht]
\centering
\includegraphics[
    width=0.9\textwidth,    
    height=5cm,      
    keepaspectratio=false 
  ]{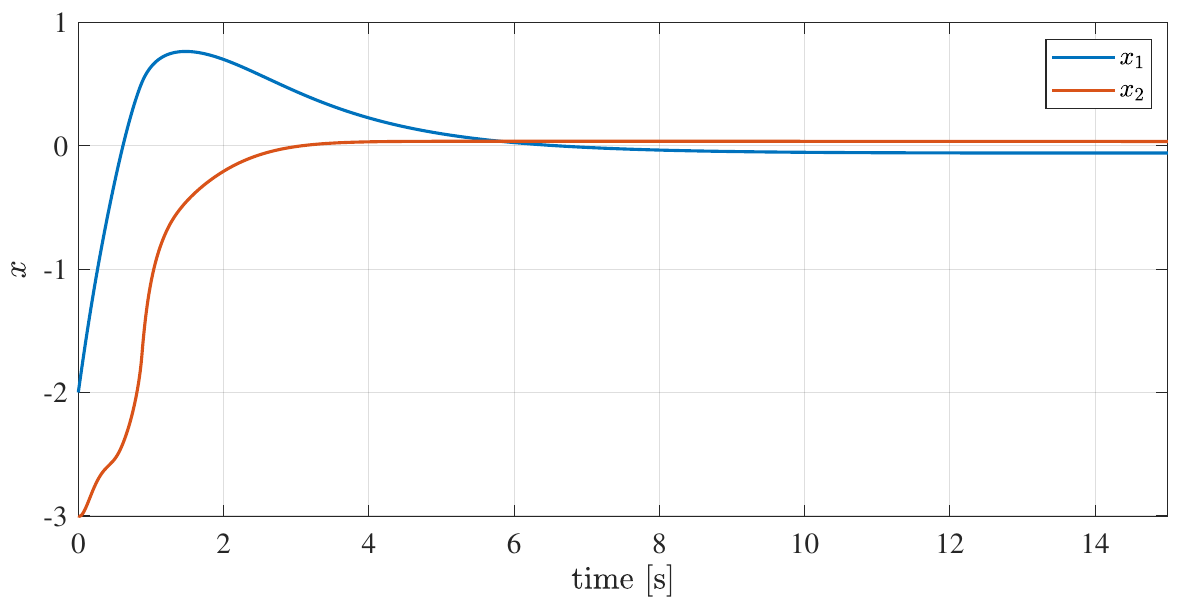}
\caption{Convergence of the state $x$}
\label{fig3}
\end{figure}

\begin{figure}[ht]
\centering
\includegraphics[
    width=0.9\textwidth,    
    height=5cm,      
    keepaspectratio=false 
  ]{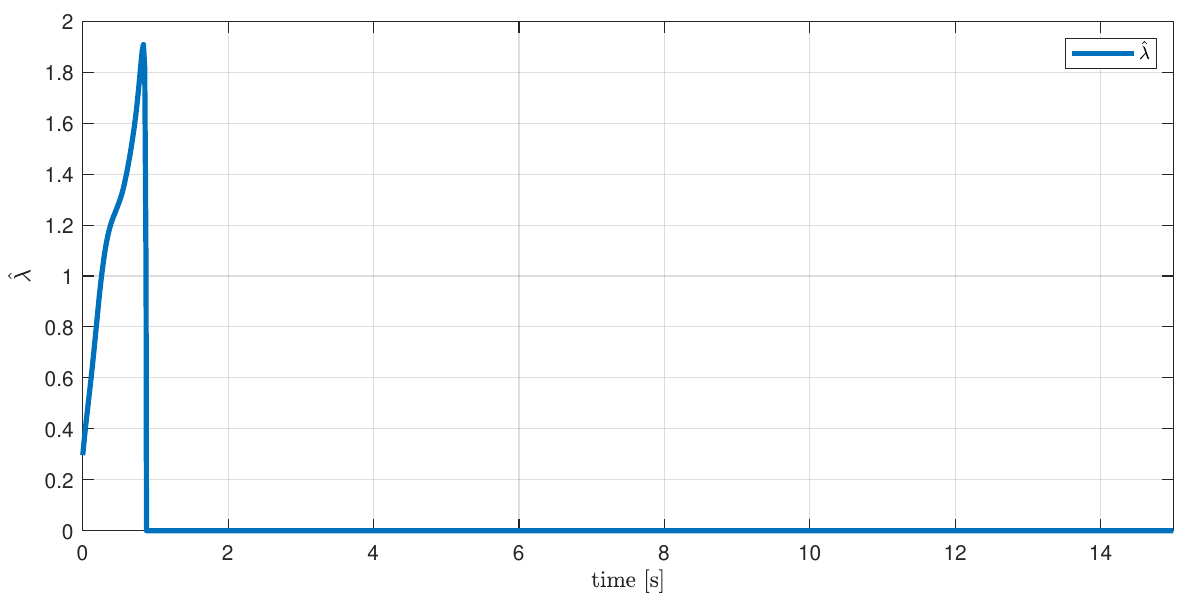}
\caption{Evolution of the estimated Lagrange multiplier $\hat{\lambda}$}
\label{fig4}
\end{figure}

Figure \ref{fig1} shows that the estimated model parameters $\hat{\theta}$ converge to their true values; Figure \ref{fig2} illustrates the evolution of the critic weights $\hat{W}$ for the approximate SVF $\hat{V}_s$; Figure \ref{fig3} demonstrates the convergence of the system states $x$ to the equilibrium point in the time domain; Figure \ref{fig4} presents the evolution of the estimated Lagrange multiplier $\hat{\lambda}$. It can be seen that when the state is close to the obstacle-like safety constraint, $\hat{\lambda}$ plays a significant role in modifying the control policy. When the state crosses the obstacle-like safety constraint, $\hat{\lambda}$ converges to zero.

\begin{figure}[ht]
\centering
\includegraphics[width=.7\columnwidth]{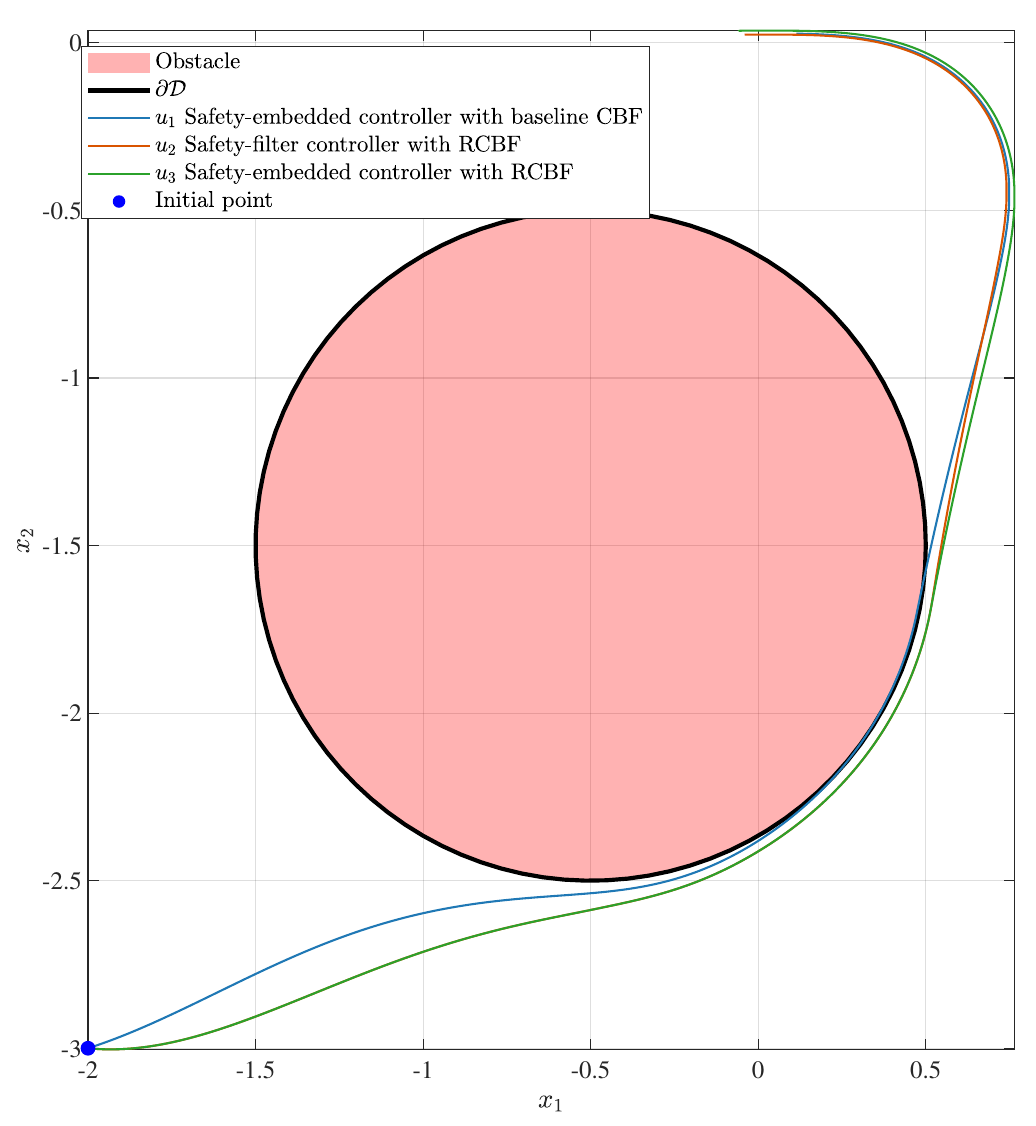}
\caption{The state trajectories under three control policies: the safety-embedded controller with the baseline CBF $u_1$; the safety filter controller with RCBF $u_2$; and the proposed safety-embedded controller with RCBF $u_3$.}
\label{fig5}
\end{figure}

To verify the improved safety assurance capability of the developed safety-embedded adaptive optimal controller, we compare the effects of three types of control policies, and Fig. \ref{fig5} presents the corresponding phase plots of the system states. The green curve represents the state trajectory under the safety-embedded control policy $u_1$ with the baseline CBF-based safety constraint (without the compensation term $\Xi(x)$). The blue curve represents the state trajectory under the adaptive optimal control policy $u_2$ with RCBF being a safety filter (projecting the control policy into the policy space determined by the RCBF constraint). The red curve represents the state trajectory under the proposed safety-embedded control policy $u_3$, which is derived from the constrained HJB equation augmented with the RCBF-based safety constraint. 

\begin{figure}[ht]
\centering
\includegraphics[
    width=0.9\textwidth,    
    height=5cm,      
    keepaspectratio=false 
  ]{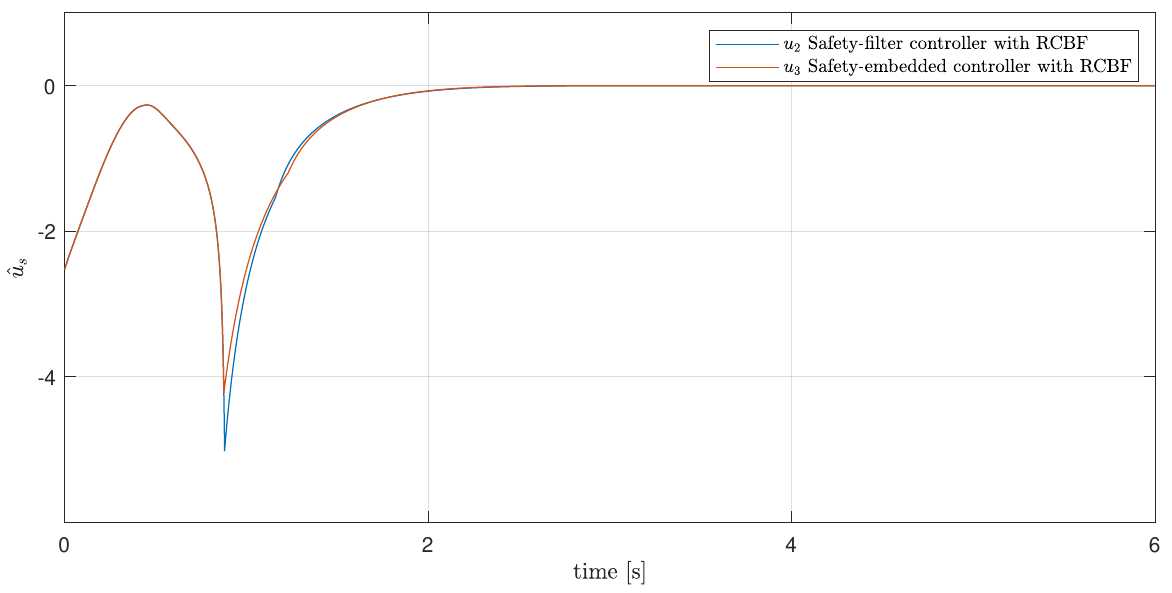}
\caption{Evolution of the controllers $u_2$ and $u_3$ during 0-6s}
\label{fig6}
\end{figure}

Comparing the controller $u_3$ we proposed with $u_1$, it can be seen that the controller $u_1$ stabilizes the system but loses the safety guarantee. This confirms the role of the compensation term in our designed RCBF, which can effectively avoid potential safety violations caused by estimation errors. The controller $u_2$ achieves the objectives of obstacle avoidance and stability in a point-wise way with the cost $V=12.18$; the controller $u_3$ achieves optimal stabilization and safety over the infinite time horizon with the cost $V=11.81$. The state trajectories corresponding to the controllers $u_2$ and $u_3$ show little distinction. Figure \ref{fig6} displays the evolution of controllers $u_2$ and $u_3$ corresponding to the phase trajectories. It can be observed that the proposed controller $u_3$ exhibits smoother variations compared to $u_2$. This is attributed to the integration of safety constraint information during the learning process, resulting in a lower cost.

\subsection{Self-triggered mechanism}
We further validate the self-triggered mechanism for implementing the adaptive safety-embedded controller with the safe set $\mathcal{D}=\{x:s(x)= -x_2^2 - x_1 +1\ge0\}$. The learning rates and initial conditions of the identifier are set as $\Gamma_{\theta} = 100\mathbb{I}_3$ and $\hat{\theta}(0) = [0,0,0]^\top$, respectively. The upper bound for the intermediate integral variable $\Omega_f$ is set as $\overline{\Omega}_f=10$. The system is initialized as $x(0) = [-3.2,-1]^\top,\hat{W}_s(0)=0.1\times \textbf{1}_{3\times 1},\Gamma(0) = 10\mathbb{I}_3$, and the learning rates are selected as $k_{c1}=0.0001,k_{c2}=0.0001,\beta=0.0001$. The self-triggered RCBF is set with $\gamma=0.8$, $\alpha_s = 6$, and the compensation term $\Xi(x)=\frac{\Vert \mathcal{L}_{\omega}s(x) \Vert^2}{1.2}$. The triggering threshold for the stability guarantee is calculated as $f_{\rm v,self}(\Vert\breve{x}_j\Vert)=\frac{\Vert \breve{x}_j\Vert}{\sqrt{10}}$, with $\chi_1=0.25$, $\chi_2 = 0.5$, and $d_v = 1$. The function $\overline{M}_j(\cdot)$ can be set as
\begin{flalign}
\overline{M}_j(\Vert \breve{e}_j\Vert)=p_1\Vert \breve{e}_j \Vert + p_2 \ln\left( 1+ \frac{p_3\Vert \breve{e}_j \Vert}{p_4\Vert \breve{x}_j\Vert + p_5} \right),
\end{flalign}
with $p_1=10$, $p_2=1$, $p_3=5$, $p_4=5$, and $p_5=10$. Hence, the triggering threshold $f_{\rm v,self}(\Vert\breve{x}_j\Vert)$ for the safety guarantee can be calculated due to \eqref{self-trigger cbf}. 

\begin{figure}[ht]
\centering
\includegraphics[
    width=0.9\textwidth,    
    height=5cm,      
    keepaspectratio=false 
  ]{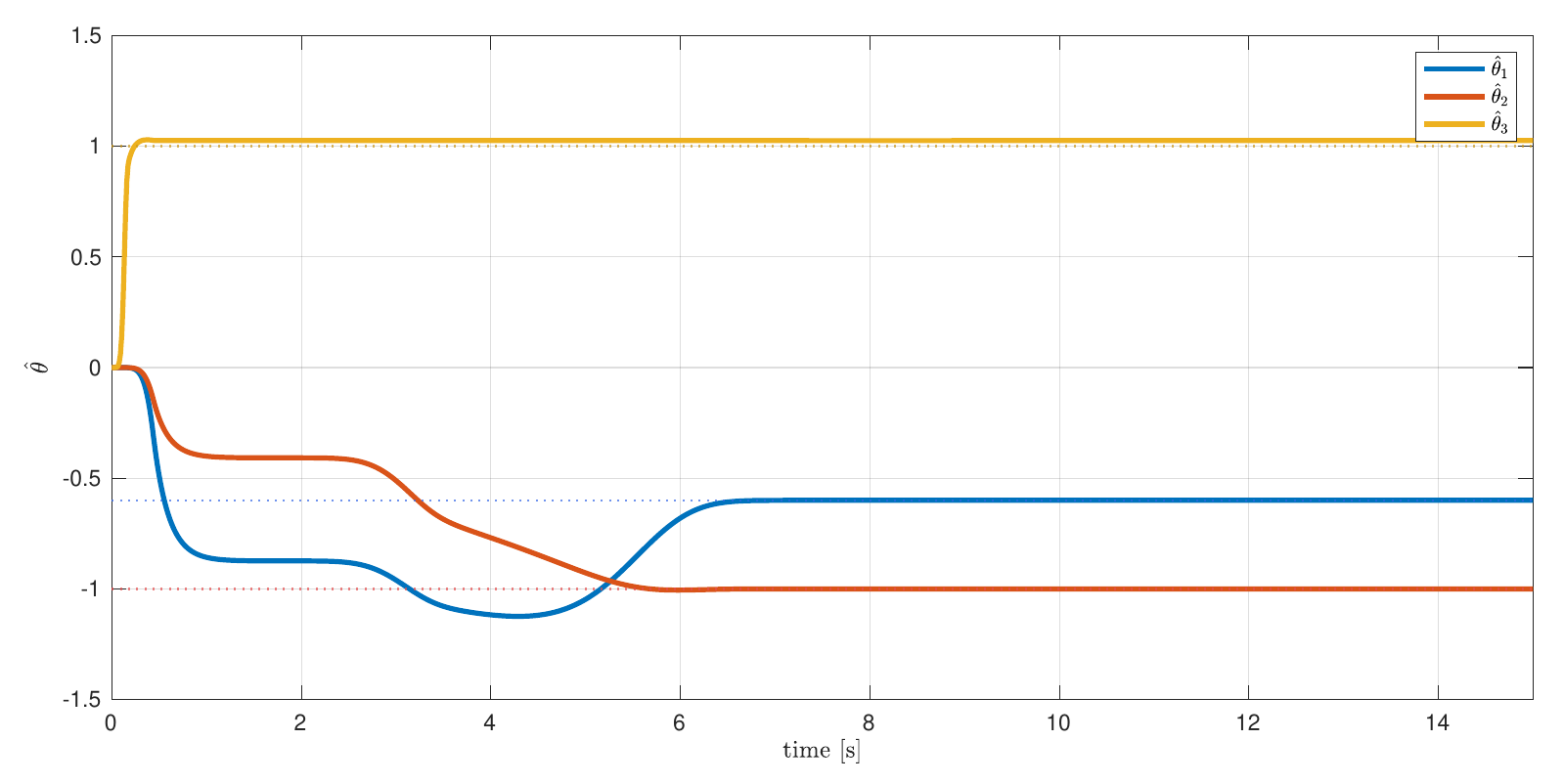}
\caption{Convergence of the estimated parameters $\hat{\theta}$}
\label{fig7}
\end{figure}

\begin{figure}[ht]
\centering
\includegraphics[
    width=0.9\textwidth,    
    height=5cm,      
    keepaspectratio=false 
  ]{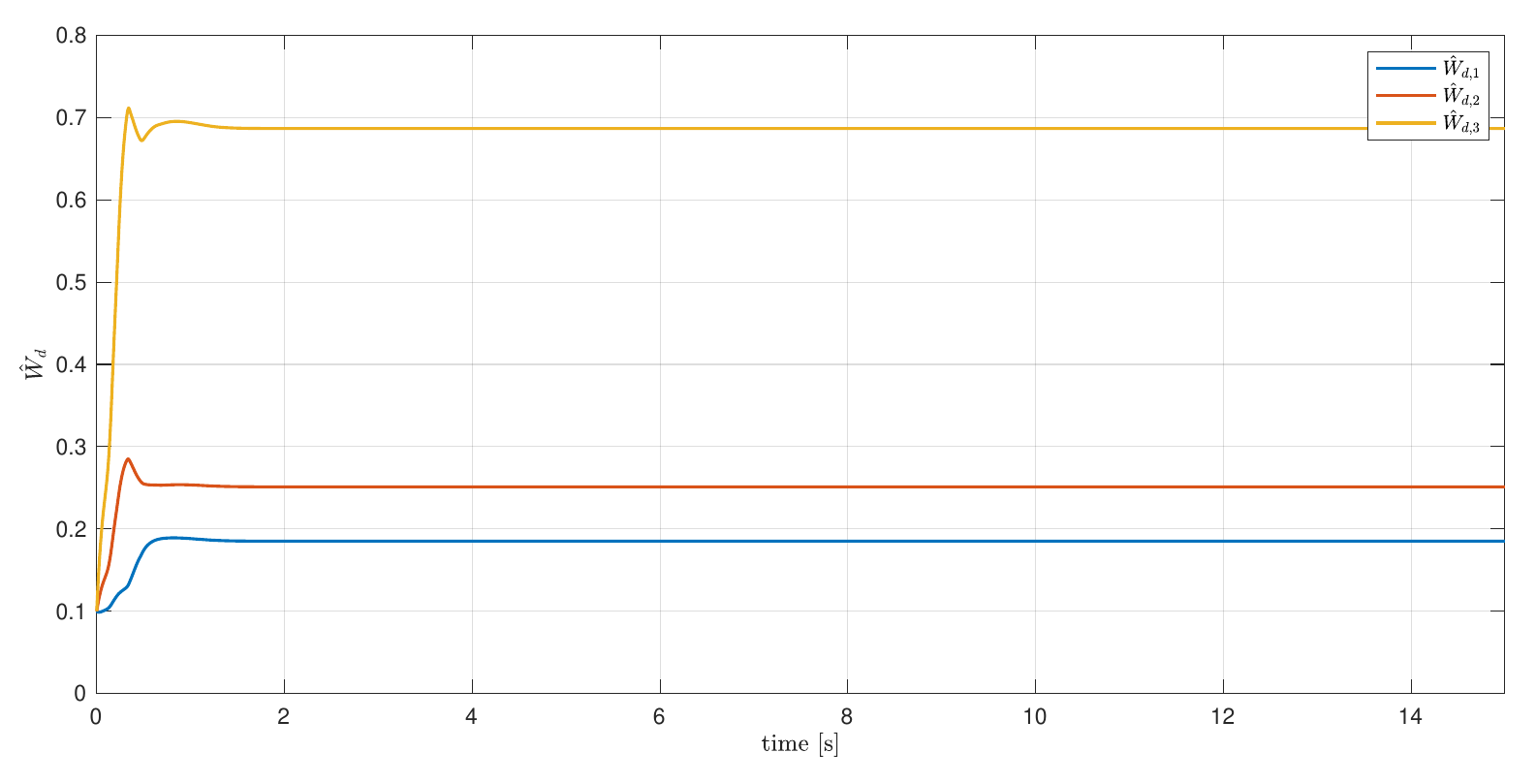}
\caption{Evolution of the approximate critic weights $\hat{W}_s$}
\label{fig8}
\end{figure}

\begin{figure}[ht]
\centering
\includegraphics[
    width=0.9\textwidth,    
    height=5cm,      
    keepaspectratio=false 
  ]{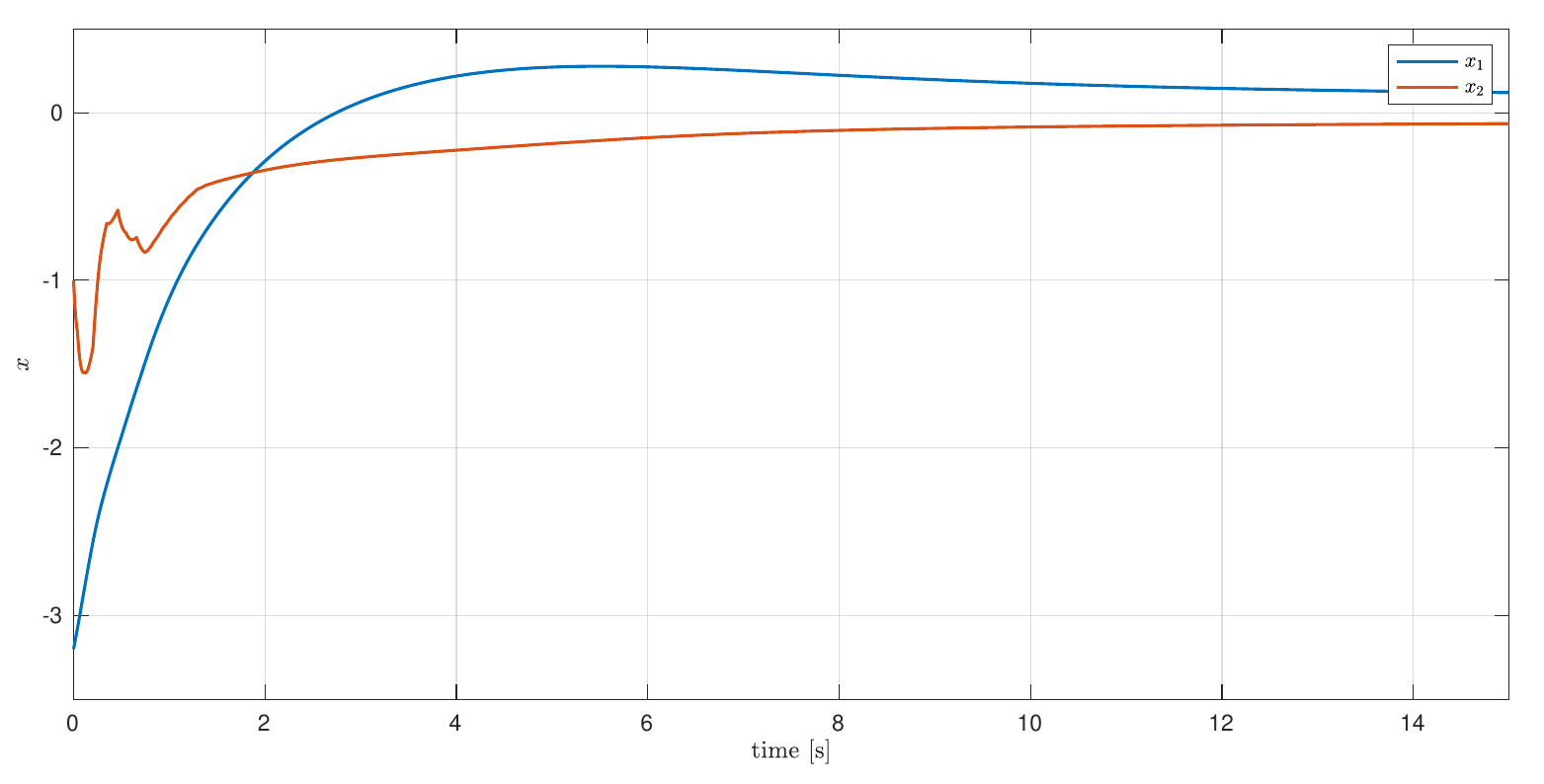}
\caption{Convergence of the true state $x$}
\label{fig9}
\end{figure}

\begin{figure}[ht]
\centering
\includegraphics[
    width=0.9\textwidth,    
    height=5cm,      
    keepaspectratio=false 
  ]{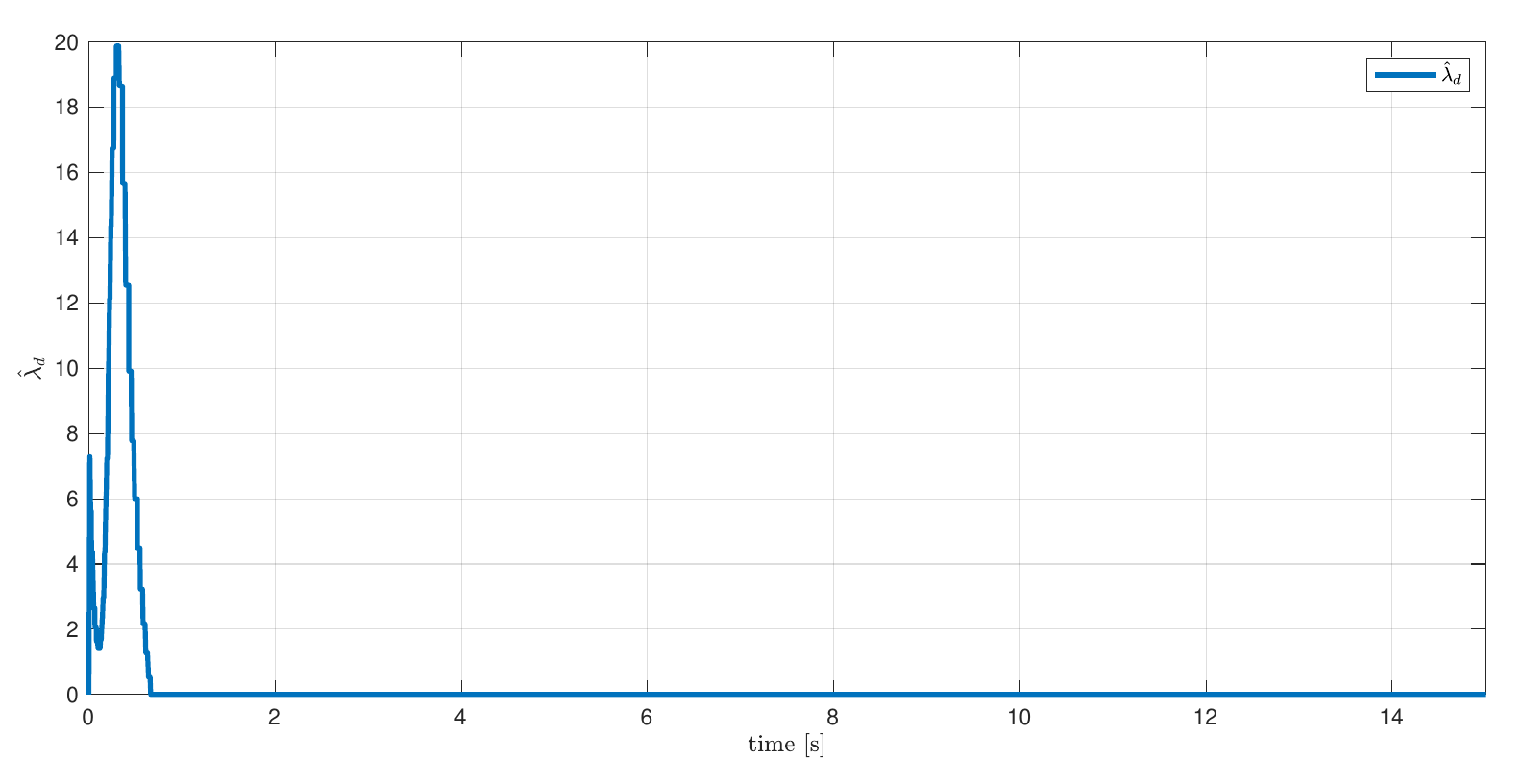}
\caption{Evolution of the estimated Lagrange multiplier $\hat{\lambda}_d$}
\label{fig10}
\end{figure}

Figure \ref{fig7} illustrates the convergence process of the estimated parameters $\hat{\theta}$ to their true values. Unlike the identification method used in the obstacle-like avoidance case, we introduce a ``refresh" mechanism that reintegrates the regressor $\Omega_f$ whenever $\Vert \Omega_f\Vert$ reaches the upper bound $\overline{\Omega}_f$. Otherwise, the estimated parameters would converge to a local optimum at t=3s. Figure \ref{fig8} shows the evolution of the estimated critic weights $\hat{W}_s$. Figure \ref{fig9} illustrates the convergence process of the true state $x$ in the time domain. Figure \ref{fig10} presents the evolution of the estimated Lagrange multiplier $\hat{\lambda}_d$, revealing that their computation process is self-triggered. By combining Figures \ref{fig8} and \ref{fig10}, it can be observed that when the Lagrange multiplier exhibits non-zero effects, they influence the transient learning process of the critic weights.

The evolution of the control signal $\hat{u}_{\lambda}$ is shown in Figure \ref{fig11} and the triggering threshold $\Vert \breve{e}_j \Vert$ in Figure \ref{fig12}. These figures demonstrate that the control input is modified only when the sampling error norm $\breve{e}_j(t)$ hits the threshold $\min\{ f_{\rm v,self}(\Vert\breve{x}_j\Vert), f_{\rm s,self}(\Vert\breve{x}_j\Vert) \}$; otherwise, it remains unchanged, thereby ensuring that the control policy is implemented efficiently. Moreover, we can see that this mechanism prevents the occurrence of the Zeno Phenomenon, where the control input would otherwise be updated infinitely often in a finite period. 

To demonstrate the effectiveness of our proposed self-triggered mechanism in reducing computational burdens, we compare it with the time-triggered approach (see Table \ref{tab:comparison}). We use the number of state samplings to represent the communication burden. Due to our time discretization step of $10^{-3}$, the time-triggered approach requires $15,000$ state samplings in total, whereas the self-triggered mechanism only needs $118$ samplings. For each sampled state, the controller update requires over $1000$ addition and multiplication operations in both approaches. Furthermore, the self-triggered mechanism introduces more than $100$ additional operations to compute the triggering condition. However, due to the significant reduction in sampling frequency, the overall computation frequency is drastically lowered. Compared to the time-triggered approach, which requires more than $10^{-7}$ computations, only a total of $154,462$ operations are needed, effectively reducing the computational and communication burden.

\begin{figure}[ht]
\centering
\includegraphics[
    width=0.9\textwidth,    
    height=5cm,      
    keepaspectratio=false 
  ]{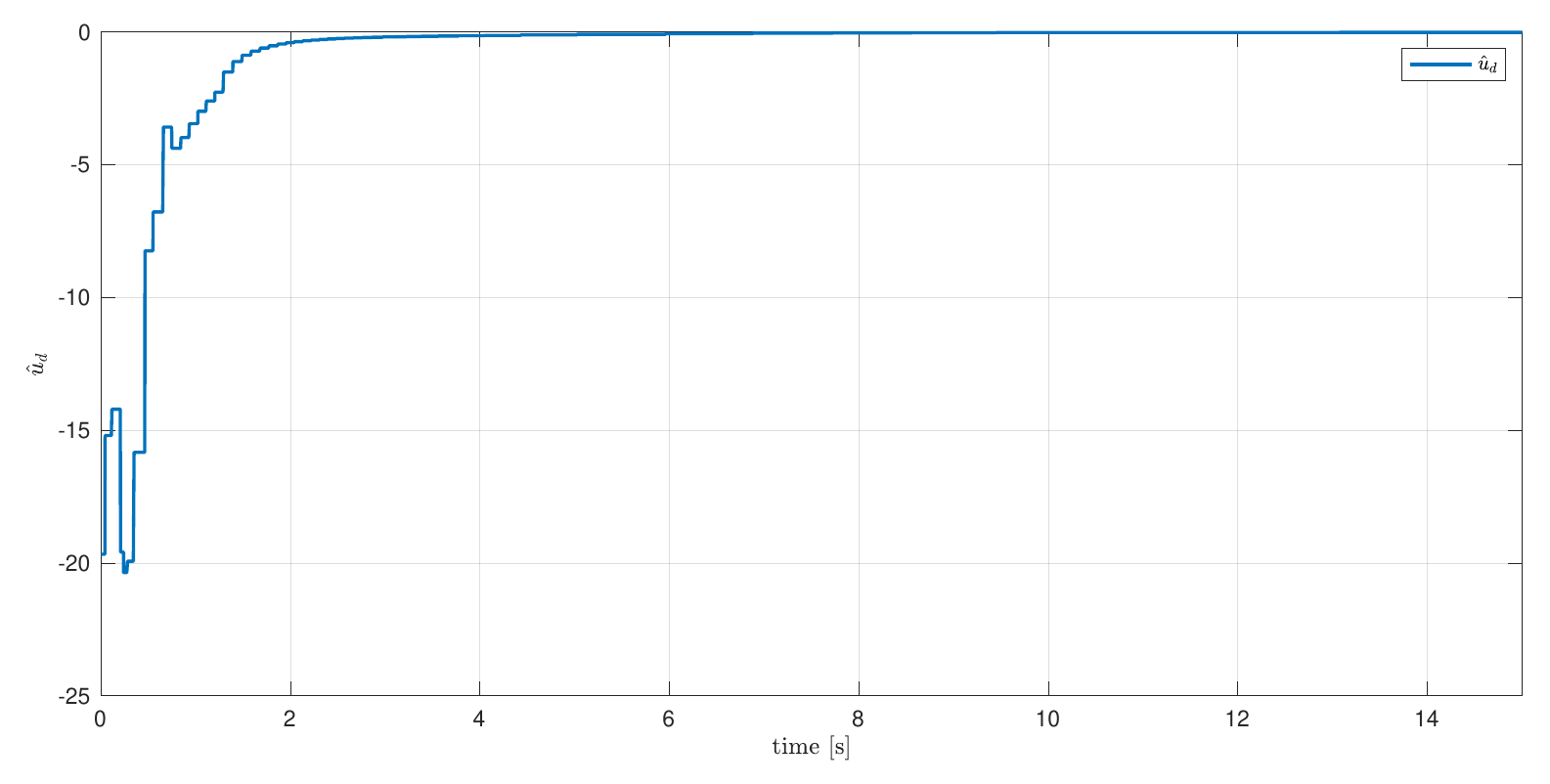}
\caption{Evolution of the self-triggered controller $\hat{u}_d$}
\label{fig11}
\end{figure}

\begin{figure}[ht]
\centering
\includegraphics[
    width=0.9\textwidth,    
    height=5cm,      
    keepaspectratio=false 
  ]{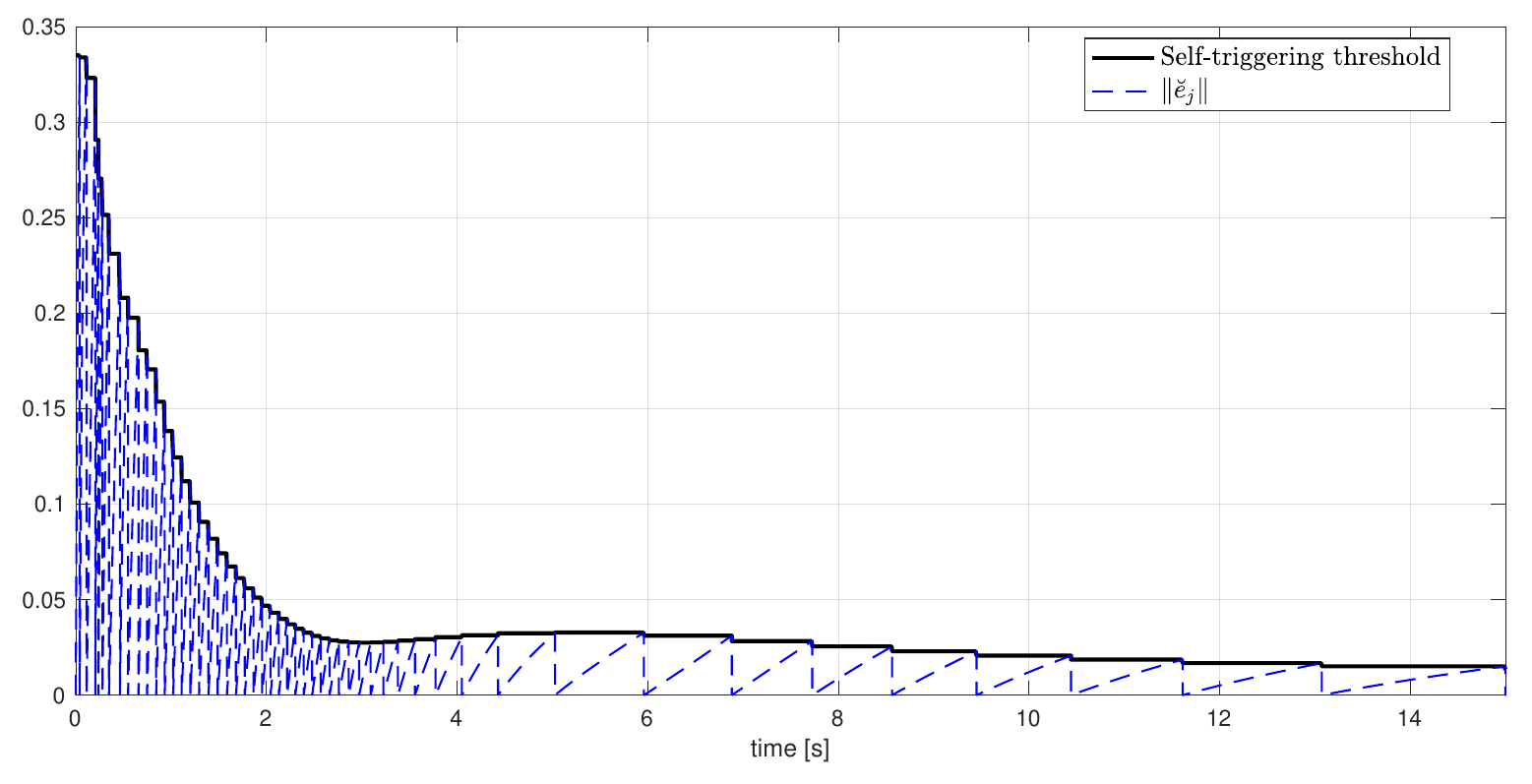}
\caption{Evolution of the self-triggering threshold and the norm of sampling error $\Vert \breve{e}_j \Vert$}
\label{fig12}
\end{figure}

\begin{table}[htbp]
\centering
\small
\caption{Comparison of communication and computational burdens: time-triggered vs.\ self-triggered mechanisms.}
\label{tab:comparison}
\setlength{\tabcolsep}{8pt}      
\renewcommand{\arraystretch}{1.2}
\begin{tabular}{@{}l|c|>{\centering\arraybackslash}p{2.8cm}|>{\centering\arraybackslash}p{2.8cm}@{}}
\toprule
\multicolumn{2}{c|}{\multirow{1}{*}[-1.6ex]{Controller implementation methods}} & Time-triggered mechanism & Self-triggered mechanism \\
\midrule
\multirow{2}{*}{\parbox{3cm}{\centering\scriptsize Number of additions \& multiplications}}
& Controller update   & $1,146$   & $1,146$  \\ \cmidrule(l){2-4}
& Triggering condition & 0    & $163$   \\
\midrule
\multicolumn{2}{c|}{Number of sampled states} & 15,000 & 118 \\
\midrule
\multicolumn{2}{c|}{Total number of computations} & $> 10^7$ & $ 154,462$ \\
\bottomrule
\end{tabular}
\end{table}

We further explore the control effects of two types of control policies as in Figure \ref{fig13}: first, the self-triggered safety-embedded control policy with baseline CBF-based safety constraints $\hat{u}_{\rm d,1}$; second, the self-triggered safety-embedded control policy with RCBF-based safety constraints $\hat{u}_{\rm d,2}$. From Figure \ref{fig13}, it can be observed that the introduction of the compensation term in RCBF effectively avoids potential safety violations caused by estimation errors. Meanwhile, the system's safety can still be guaranteed under the self-triggered mechanism.

\begin{figure}[ht]
\centering
\includegraphics[width=.7\columnwidth]{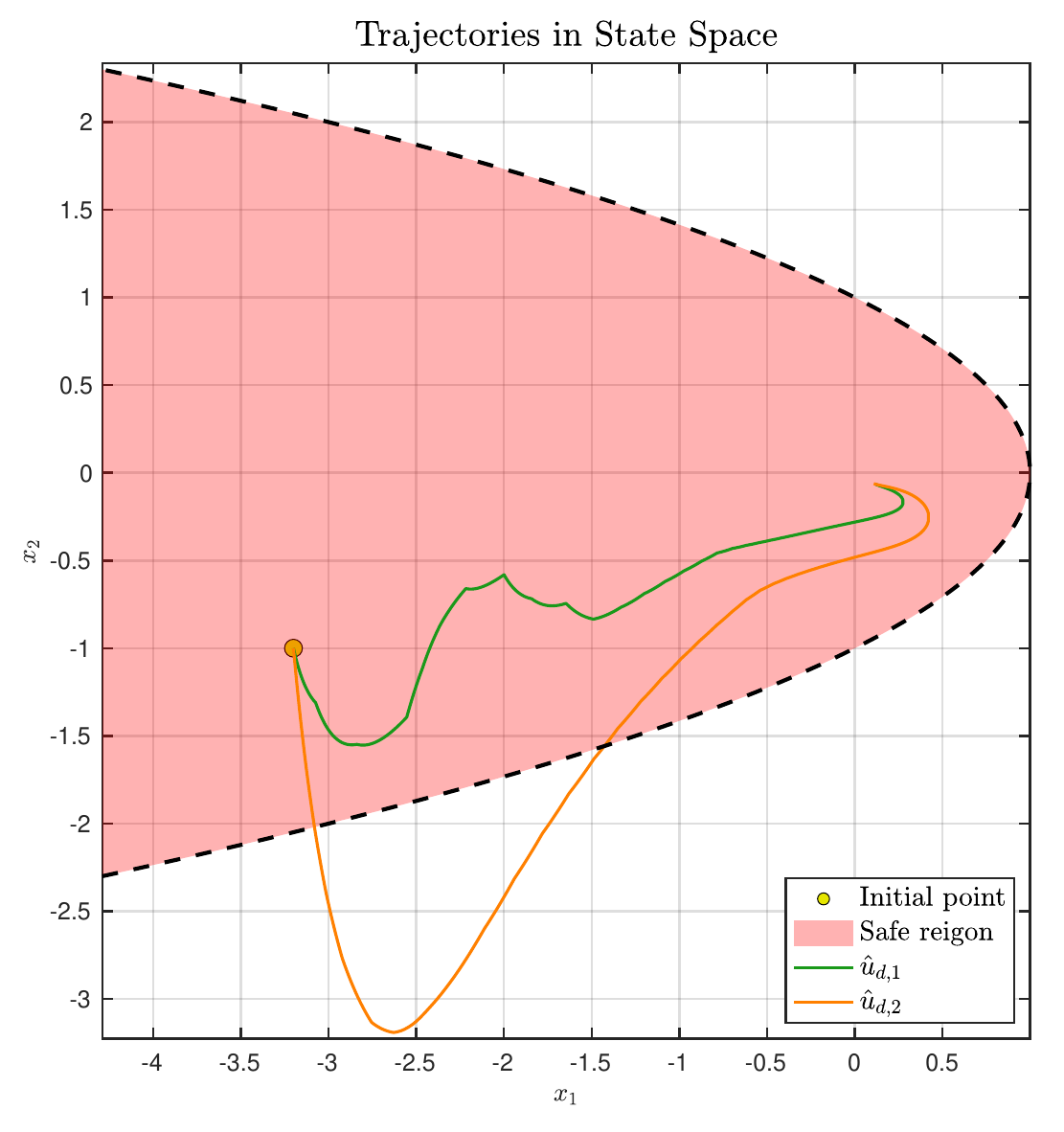}
\caption{The state trajectories under two control policies: the self-triggered safety-embedded controller with the baseline CBF $\hat{u}_{\rm d,1}$; the proposed self-triggered safety-embedded controller with RCBF $\hat{u}_{\rm d,2}$.}
\label{fig13}
\end{figure}

\section{Conclusion and future work}
In this paper, we focus on the learning-based optimal control design for state-constrained systems with unknown drift dynamics. To ensure the safety of such systems, we develop an identifier for estimating uncertain parameters and construct an RCBF to encode the safety guarantee for uncertain systems. Then, the RCBF-based safety constraint is embedded in the value function via a Lagrange multiplier to guarantee optimal stabilization and safety simultaneously, and we develop a safety-embedded critic learning framework to solve the RCBF-constrained HJB equation online. Additionally, a self-triggered mechanism is proposed to implement the learning-based controller for reducing the computational and communication overhead. Correspondingly, self-triggered RCBF-constrained HJB equations are derived and can be solved similarly through the safety-embedded critic learning. In the future, this safety-embedded adaptive optimal control method will be extended to systems with incomplete state information, where an observer-controller synthesis will be considered.

\appendix
\section{Proof for triggering interval \eqref{safe period}} \label{appendix2}
For $t\in[t_j,t_{j+1})$, the norm of the derivative of $x$ is bounded by
\begin{flalign}\label{uu}
\Vert \dot{x} \Vert &\leq \Vert \zeta(x,\theta) \Vert + \Vert \rho(x) \Vert \Vert u(\breve{x}_j) \Vert \notag\\
&\leq d_{\zeta} \Vert x \Vert + \Vert \rho(x)\Vert \left( \Vert u(x) - u(\breve{x}_j) \Vert + \Vert u(x)\Vert \right)\notag\\
&\leq d_{\zeta} \Vert x \Vert + \Vert \rho(x) \Vert\left( d_v \Vert \breve{e}_j(t) \Vert + \Vert u(x) \Vert \right), 
\end{flalign}
where $d_{\zeta}$ and $d_v$ are Lipschitz constants of $\zeta(\cdot,\theta)$ and $u(\cdot)$. Suppose that the controller $u(\cdot)$ is bounded, we can rewrite \eqref{uu} as
\begin{flalign}
& \Vert \dot{\breve{e}}_j(t) \Vert = \Vert \dot{x} \Vert \leq l_1 \Vert x \Vert +  l_2 \Vert \breve{e}_j(t) \Vert + l_3, 
\end{flalign}
where $l_1 = d_{\zeta}$, $l_2 = d_v \overline{\Vert \rho(x)\Vert}$ and $l_3 = \overline{\Vert \rho(x)\Vert} u_{\max}$. Further, we can obtain that 
\begin{flalign}
& \frac{d\Vert \breve{e}_j(t) \Vert}{dt} \leq \Vert \dot{\breve{e}}_j(t) \Vert \leq l_1 \Vert x_j \Vert + (l_1+l_2)\Vert \breve{e}_j(t) \Vert + l_3. 
\end{flalign}
Due to comparing lemma, we can obtain that
\begin{flalign}
& \Vert \breve{e}_j(t) \Vert \leq \frac{l_1\Vert \breve{x}_j \Vert+l_3}{l_1+l_2}\left( {\rm exp}\{(l_1+l_2)(t-t_j)\}-1 \right). 
\end{flalign}
If the triggering threshold $e_{j,b}$ is violated, i.e., $t=t_{j+1}$, we have
\begin{flalign}\label{asd}
& \frac{l_1\Vert \breve{x}_j \Vert+l_3}{l_1+l_2}\left( {\rm exp}\{(l_1+l_2)T_j\}-1 \right)\ge \Vert \breve{e}(t_{j+1}) \Vert = e_{j,b}. 
\end{flalign}
Following (\ref{asd}), we have
\begin{flalign}\label{bound for Tj}
& T_j > \frac{1}{l_1+l_2}\ln\left( 1+ \frac{l_1+l_2}{l_1\Vert \breve{x}_j\Vert+l_3}e_{j,b} \right)\triangleq T_{j,b}(\breve{x}_j). 
\end{flalign}


\bibliographystyle{elsarticle-num} 
\bibliography{arXiv}

\begin{thebibliography}{10}
\expandafter\ifx\csname url\endcsname\relax
  \def\url#1{\texttt{#1}}\fi
\expandafter\ifx\csname urlprefix\endcsname\relax\def\urlprefix{URL }\fi
\expandafter\ifx\csname href\endcsname\relax
  \def\href#1#2{#2} \def\path#1{#1}\fi

\bibitem{hsu2023safety}
K.-C. Hsu, H.~Hu, J.~F. Fisac, The safety filter: A unified view of safety-critical control in autonomous systems, Annual Review of Control, Robotics, and Autonomous Systems 7 (2023).

\bibitem{lewis2013reinforcement}
F.~L. Lewis, D.~Liu, Reinforcement learning and approximate dynamic programming for feedback control, John Wiley \& Sons, 2013.

\bibitem{vamvoudakis2010online}
K.~G. Vamvoudakis, F.~L. Lewis, Online actor--critic algorithm to solve the continuous-time infinite horizon optimal control problem, Automatica 46~(5) (2010) 878--888.

\bibitem{wang2017adaptive}
D.~Wang, H.~He, D.~Liu, Adaptive critic nonlinear robust control: A survey, IEEE transactions on cybernetics 47~(10) (2017) 3429--3451.

\bibitem{yang2020event}
X.~Yang, H.~He, Event-driven h$\infty$-constrained control using adaptive critic learning, IEEE Transactions on Cybernetics 51~(10) (2020) 4860--4872.

\bibitem{na2020adaptive}
J.~Na, Y.~Lv, K.~Zhang, J.~Zhao, Adaptive identifier-critic-based optimal tracking control for nonlinear systems with experimental validation, IEEE Transactions on Systems, Man, and Cybernetics: Systems 52~(1) (2020) 459--472.

\bibitem{cohen2020approximate}
M.~H. Cohen, C.~Belta, Approximate optimal control for safety-critical systems with control barrier functions, in: 2020 59th IEEE conference on decision and control (CDC), IEEE, 2020, pp. 2062--2067.

\bibitem{greene2020sparse}
M.~L. Greene, P.~Deptula, S.~Nivison, W.~E. Dixon, Sparse learning-based approximate dynamic programming with barrier constraints, IEEE Control Systems Letters 4~(3) (2020) 743--748.

\bibitem{ames2016control}
A.~D. Ames, X.~Xu, J.~W. Grizzle, P.~Tabuada, Control barrier function based quadratic programs for safety critical systems, IEEE Transactions on Automatic Control 62~(8) (2016) 3861--3876.

\bibitem{cheng2019end}
R.~Cheng, G.~Orosz, R.~M. Murray, J.~W. Burdick, End-to-end safe reinforcement learning through barrier functions for safety-critical continuous control tasks, in: Proceedings of the AAAI conference on artificial intelligence, Vol.~33, 2019, pp. 3387--3395.

\bibitem{peng2023design}
C.~Peng, X.~Liu, J.~Ma, Design of safe optimal guidance with obstacle avoidance using control barrier function-based actor--critic reinforcement learning, IEEE Transactions on Systems, Man, and Cybernetics: Systems (2023).

\bibitem{krstic2023inverse}
M.~Krstic, Inverse optimal safety filters, IEEE Transactions on Automatic Control 69~(1) (2023) 16--31.

\bibitem{cohen2023safe}
M.~H. Cohen, C.~Belta, Safe exploration in model-based reinforcement learning using control barrier functions, Automatica 147 (2023) 110684.

\bibitem{almubarak2021hjb}
H.~Almubarak, E.~A. Theodorou, N.~Sadegh, Hjb based optimal safe control using control barrier functions, in: 2021 60th IEEE Conference on Decision and Control (CDC), IEEE, 2021, pp. 6829--6834.

\bibitem{wang2025learning}
X.~Wang, H.~Zhang, S.~Wang, W.~Xiao, M.~Guay, Learning-enhanced safeguard control for high-relative-degree systems: Robust optimization under disturbances and faults, arXiv preprint arXiv:2501.15373 (2025).

\bibitem{bandyopadhyay2023hjb}
S.~Bandyopadhyay, S.~Bhasin, Hjb based online safe reinforcement learning for state-constrained systems, arXiv preprint arXiv:2305.12967 (2023).

\bibitem{vamvoudakis2014event}
K.~G. Vamvoudakis, Event-triggered optimal adaptive control algorithm for continuous-time nonlinear systems, IEEE/CAA Journal of Automatica Sinica 1~(3) (2014) 282--293.

\bibitem{vamvoudakis2018model}
K.~G. Vamvoudakis, H.~Ferraz, Model-free event-triggered control algorithm for continuous-time linear systems with optimal performance, Automatica 87 (2018) 412--420.

\bibitem{breeden2021control}
J.~Breeden, K.~Garg, D.~Panagou, Control barrier functions in sampled-data systems, IEEE Control Systems Letters 6 (2021) 367--372.

\bibitem{sun2024safety}
J.~Sun, J.~Yang, Z.~Zeng, Safety-critical control with control barrier function based on disturbance observer, IEEE Transactions on Automatic Control (2024).

\bibitem{yang2019self}
G.~Yang, C.~Belta, R.~Tron, Self-triggered control for safety critical systems using control barrier functions, in: 2019 American control conference (ACC), IEEE, 2019, pp. 4454--4459.

\bibitem{xiao2022event}
W.~Xiao, C.~Belta, C.~G. Cassandras, Event-triggered control for safety-critical systems with unknown dynamics, IEEE Transactions on Automatic Control 68~(7) (2022) 4143--4158.

\bibitem{sabouni2024optimal}
E.~Sabouni, C.~G. Cassandras, W.~Xiao, N.~Meskin, Optimal control of connected automated vehicles with event/self-triggered control barrier functions, Automatica 162 (2024) 111530.

\bibitem{kolathaya2018input}
S.~Kolathaya, A.~D. Ames, Input-to-state safety with control barrier functions, IEEE control systems letters 3~(1) (2018) 108--113.

\bibitem{taylor2020adaptive}
A.~J. Taylor, A.~D. Ames, Adaptive safety with control barrier functions, in: 2020 American Control Conference (ACC), IEEE, 2020, pp. 1399--1405.

\bibitem{kirk2004optimal}
D.~E. Kirk, Optimal control theory: an introduction, Courier Corporation, 2004.

\bibitem{hassan2002nonlinear}
K.~K. Hassan, et~al., Nonlinear systems, Departement of Electrical and computer Engineering, Michigan State University (2002).

\bibitem{ming2022self}
Z.~Ming, H.~Zhang, Y.~Yan, J.~Sun, Self-triggered adaptive dynamic programming for model-free nonlinear systems via generalized fuzzy hyperbolic model, IEEE Transactions on Systems, Man, and Cybernetics: Systems 53~(5) (2022) 2792--2801.

\bibitem{mahmud2021safe}
S.~N. Mahmud, S.~A. Nivison, Z.~I. Bell, R.~Kamalapurkar, Safe model-based reinforcement learning for systems with parametric uncertainties, Frontiers in Robotics and AI 8 (2021) 733104.

\bibitem{kamalapurkar2016efficient}
R.~Kamalapurkar, J.~A. Rosenfeld, W.~E. Dixon, Efficient model-based reinforcement learning for approximate online optimal control, Automatica 74 (2016) 247--258.

\bibitem{deptula2021approximate}
P.~Deptula, Z.~I. Bell, F.~M. Zegers, R.~A. Licitra, W.~E. Dixon, Approximate optimal influence over an agent through an uncertain interaction dynamic, Automatica 134 (2021) 109913.

\bibitem{yang2019adaptive}
X.~Yang, H.~He, Adaptive critic learning and experience replay for decentralized event-triggered control of nonlinear interconnected systems, IEEE Transactions on Systems, Man, and Cybernetics: Systems 50~(11) (2019) 4043--4055.

\bibitem{jankovic2018robust}
M.~Jankovic, Robust control barrier functions for constrained stabilization of nonlinear systems, Automatica 96 (2018) 359--367.

\end{thebibliography}



\end{document}